\documentclass[letterpaper,11pt]{article}

\usepackage{style}

\title{\bf Equational Bit-Vector Solving via Strong Gr\"{o}bner Bases}

\author[1,2]{Jiaxin Song \thanks{Email: jsongbk@ust.hk}}
\author[1]{Hongfei Fu \thanks{Email: jt002845@sjtu.edu.cn} \thanks{Hongfei is the corresponding author.}}
\author[2]{Charles Zhang \thanks{Email: charlesz@cse.ust.hk}}
\affil[1]{\it Shanghai Jiao Tong University, China}
\affil[2]{\it Hong Kong University of Science and Technology, China}
\date{}

\begin{document}
\maketitle

\begin{abstract}
Bit-vectors, which are integers in a finite number of bits, are ubiquitous in software and hardware systems. In this work, we consider the satisfiability modulo theories (SMT) of bit-vectors. 
Unlike normal integers, the arithmetics of bit-vectors are modular upon integer overflow. Therefore, the SMT solving of bit-vectors needs to resolve the underlying modular arithmetics.  
In the literature, two prominent approaches for SMT solving are bit-blasting (that transforms the SMT problem into boolean satisfiability) and integer solving (that transforms the SMT problem into integer properties).
Both approaches ignore the algebraic properties of the modular arithmetics 
and hence could not utilize these properties to improve the efficiency of SMT solving. 

In this work, we consider the equational theory of bit-vectors and capture the algebraic properties behind them via strong Gr\"{o}bner bases.
First, we apply strong Gr\"{o}bner bases to the quantifier-free equational theory of bit-vectors and propose a novel algorithmic improvement in the key computation of multiplicative inverse modulo a power of two. 
Second, we resolve the important case of invariant generation in quantified equational bit-vector properties via strong Gr\"{o}bner bases and linear congruence solving. 
Experimental results over an extensive range of benchmarks show that our approach outperforms existing methods in both time efficiency and memory consumption. 
\end{abstract}

\section{Introduction}
In software and hardware systems, integers are often represented by a finite number of bits, resulting in \emph{bit-vectors} (or \emph{machine integers}) that take values from a bounded range. 
Unlike normal integer arithmetics, the integer overflow in bit-vectors is often handled via modular arithmetics. 
This causes a significant problem in verifying program correctness with bit-vectors. 
For example, the conditional branch 
\[
\texttt{if (x > 0 \&\& y > 0)\, assert(x + y > 0);}
\]
is unsafe since the value of \texttt{\small x + y} may overflow, and simply treating the variables 
\texttt{\small x} and \texttt{\small y} as unbounded integers would produce incorrect results. 
Thus, integer overflow imposes a challenge to verify the correctness of bit-vector programs.

As bit-vectors are rudimentary, the correctness of software and hardware systems heavily relies on the correctness of bit-vector operations. 
Hence, verification of bit-vector properties has received significant attention in the literatures~\cite{DBLP:journals/fmsd/WintersteigerHM13,DBLP:conf/fmcad/Griggio11,DBLP:journals/toplas/ElderLSAR14,DBLP:journals/toplas/ChenDKSW18}. 
An important subject in the verification of bit-vectors is their satisfiability modulo theories (SMT)~\cite{DBLP:journals/jar/BeyerDW18,DBLP:series/txtcs/KroeningS16} that aim to solve the satisfiability of formulas over bit-vectors. 
In the SMT of bit-vector theory, the formulas of concern usually include standard operations such as addition, multiplication, bitwise-or/and, division, signed/unsigned comparison, etc.

By distinguishing whether the overflow of bit-vectors is discarded immediately or recorded by extra bits, bit-vectors can be of either \emph{fixed} or \emph{flexible} size. 
Fixed-size bit-vectors (see~\cite{smtlibv26} and \cite[Chapter 6]{DBLP:series/txtcs/KroeningS16}) completely throw away the overflow bit, and hence the arithmetics are exactly the modular arithmetics. 
Flexible-size bit-vectors~\cite{DBLP:conf/sat/FuhsGMSTZ07,DBLP:conf/lpar/ZanklM10,DBLP:conf/issta/JiaH00MZ23} record the overflow by specialized bits, and hence achieve varying size for a bit-vector. 

In the literature, there are two prominent approaches to solving the SMT of bit-vectors. 
The first approach is often called \emph{bit-blasting} \cite[Chapter 6]{DBLP:series/txtcs/KroeningS16} that transforms a bit-vector equivalently into the collection of bits in the bit-vector and solves the SMT via boolean satisfiability. 
Bit-blasting has the drawback that it completely breaks the algebraic structure behind bit-vector arithmetics and hence cannot utilize the algebraic properties to improve SMT solving. 
To resolve this drawback, the second approach~\cite{DBLP:conf/fmcad/Griggio11,DBLP:conf/vmcai/Jovanovic17,DBLP:conf/smt/Graham-Lengrand17} transforms bit-vector arithmetics into integer arithmetics and solves the original formula via integer SMT solving such as linear and polynomial integer arithmetics. 
Its main obstacle is SMT solving of the polynomial theory of integers which is highly difficult to handle.

To fully capture the algebraic structure of bit-vectors, Gr\"{o}bner bases have been applied to handle the equational theory of bit-vectors. 
Note that the classical Gr\"{o}bner basis~\cite[Chapter 1]{groebnerbasis} is limited to polynomials with coefficients from a field and cannot be applied to polynomials with coefficients from the ring $\mathbb{Z}_{2^d}$ ($d>1$) of bit-vectors, since the ring $\mathbb{Z}_{2^d}$ is a field only when $d=1$. 
To circumvent this issue, several approaches~\cite{DBLP:conf/fmcad/KaufmannBK19,JPAAgroebnerbasis,DBLP:conf/sat/SeedKE20} have considered extensions of Gr\"{o}bner bases.
The work~\cite{DBLP:conf/fmcad/KaufmannBK19} considers Gr\"{o}bner bases with coefficients from a principal ideal domain (in particular, the set of integers)~\cite[Chapter 4]{groebnerbasis}. 
The work~\cite{JPAAgroebnerbasis} establishes the general theory of Gr\"{o}bner bases over polynomials with coefficients from a commutative Noetherian ring. 
The work~\cite{DBLP:conf/sat/SeedKE20} further improves the computation of Gr\"{o}bner bases in ~\cite{JPAAgroebnerbasis} by a heuristics. 
These approaches are either too narrow (e.g., only considering principal ideal domains) or too wide (that consider general commutative Noetherian rings), resulting in excessive computations in Gr\"{o}bner bases for bit-vectors. 

In this work, we consider the SMT of fixed-size bit-vectors. We focus on the SMT solving of the theory of polynomial equations over bit-vectors, which is the basic class of SMT that considers polynomial (in)equations with modular addition and multiplication, 
and finds applications in verification of arithmetic circuits~\cite{DBLP:conf/fmcad/KaufmannBK19,DBLP:conf/cav/WienandWSKG08}. 
This work aims to develop novel SMT-solving algorithms that leverage the algebraic structure from the ring of bit-vectors to improve the efficiency of SMT solving. 
Our detailed contributions are as follows:
\vspace{-1.5mm}
\begin{itemize}
\item First, we propose a novel approach to solve the quantifier-free polynomial equational theory of bit-vectors via strong Gr\"{o}bner bases~\cite{norton2001strong}.
Strong Gr\"{o}bner bases extend Gr\"{o}bner bases to polynomials with coefficients from a principal ideal ring (PIR) and, therefore well fits bit-vectors. A key contribution here is a theorem that establishes the connection between the existence of a constant polynomial in an ideal and in a strong Gr\"{o}bner basis for the ideal. 

\item Second, we propose an algorithmic improvement for the key calculation of multiplicative inverse modulo a power of two in the computation of strong Gr\"{o}bner bases. 
As multiplicative inverse is a main factor in computing strong Gr\"{o}bner bases, the improvement substantially speeds up SMT solving.  
\item Third, we propose an invariant generation method for bit-vectors. Note that invariant generation can be encoded as a special class of constraint Horn clauses (CHC) and is an important case of quantified SMT solving~\cite{DBLP:conf/kbse/YaoKSFWR23}. 
We show that the generation of polynomial equational invariants can be solved by strong Gr\"{o}bner bases and linear congruence solving. 
\end{itemize}

We implement our approach in the \textsf{cvc5} SMT solver~\cite{barbosa2022cvc5}. Experimental results over a wide range of benchmarks show that our approach substantially outperforms existing approaches in both the number of solved instances, time efficiency, and memory consumption. 
Especially, for quantifier-free SMT, our method can solve $40\%$ more unsatisfiable instances compared to state-of-the-art approaches. 
For polynomial invariant generation, our method achieves a $20$X speedup and a $6$X reduction in memory usage.

\section{Preliminaries}
\label{sec:background}

In this section, we present basic concepts in rings, polynomials, strong Gr\"{o}bner bases, and bit-vectors. We refer to standard textbooks (e.g.,~\cite{algebra}) for a detailed treatment of rings and polynomials. 

\subsection{Rings, Polynomials and the Ring of Bit-Vectors}
Generally, a \emph{ring} is a non-empty set $R$ equipped with two binary operations $+,\cdot: R\times R\rightarrow R$ (where $+$ is the abstract addition and $\cdot$ is the abstract multiplication) and two constants $0$ and $1$, such that (i) the addition $+$ is commutative and associative, (ii) the multiplication $\cdot$ is associative and distributive over the addition, (iii) the element $0$ is the identity element for the addition, and (iv) the element $1$ is the identity element for the multiplication. Formally, a \emph{ring} is an algebraic structure $(R, +, \cdot, 0, 1)$ such that $R$ is a non-empty set,  $+,\cdot$ are functions from $R\times R$ into $R$, $0,1\in R$ are constants in $R$, and the following properties hold for all $a,b,c\in R$:

\begin{itemize}
\item $a+b=b+a$, $(a+b)+c=a+(b+c)$ and $(a\cdot b)\cdot c=a\cdot (b\cdot c)$;
\item $a\cdot (b+c)=(a\cdot b) + (a\cdot c)$ and $(b+c)\cdot a=(b\cdot a) + (c\cdot a)$;
\item $a+0=a$ and $a\cdot 1=1\cdot a=a$. 
\item there is a unique element $d\in R$ (usually denoted by $-a$) such that $a+d=0$. 
\end{itemize}
A ring $(R, +, \cdot, 0, 1)$ is a \emph{field} if the multiplication is commutative and for every $a\in R$, there exists $b\in R$ (called the \emph{multiplicative inverse} of $a$) such that $a\cdot b=b\cdot a=1$. The ring $(R, +, \cdot, 0, 1)$ is  \emph{communicative} if $a\cdot b = b\cdot a$ for all $a, b\in R$. 
In the following, we always consider commutative rings when referring to a ring and only write $R$ instead of  $(R, +, \cdot, 0, 1)$ for the sake of brevity. 
A subset of $R$ is called a \emph{subring} of $R$ if it contains $1$ and is closed under the ring operations induced from $R$.

A typical example of a finite ring is the ring $\mathbb{Z}_m$ ($m$ is a positive integer) of modular addition and multiplication w.r.t the modulus $m$. In this work, we pay special attention to the ring $\mathbb{Z}_{p^k}$ where $p$ is prime and $k$ is a positive integer. Notice that when $k>1$, $\mathbb{Z}_{p^k}$ is not a field as multiplicative inverse may not exist. When $p=2$, $\mathbb{Z}_{2^k}$ is exactly the ring of bit-vectors with size $k$. 
In the ring $\mathbb{Z}_{2^k}$, for any element $a\in \mathbb{Z}_{2^k}\setminus \{0\}$, we define $\nu_2(a)$ to be the maximum power of two that divides $a$, i.e., $a = 2^{\nu_2(a)}\cdot b$ and $2{\not\vert}\,b$ for some integer $b$.

A polynomial with coefficients from a ring is a finite summation of terms for which a term is a product between an element in the ring (as the coefficient) and the variables in the polynomial. Formally,  
let $R$ be a ring, $x_1, \ldots, x_n$ be $n$ variables and $\bm{x}:=\myangle{x_1,\ldots, x_n}$. 
A \emph{monomial} is a product of the form $\bm{x}^{\bm \alpha} := x_1^{\alpha_1}\cdot \ldots \cdot x_n^{\alpha_n}$, where $\bm{\alpha} = \myangle{\alpha_1, \ldots, \alpha_n}\in \mathbb{N}^n$ is a vector of natural numbers for which each $\alpha_i$ specifies the exponent of the variable $x_i$.
We define $\abs{\bm{\alpha}} := \alpha_1 + \ldots + \alpha_n$ as the \emph{degree} of the monomial $\bm{x}^{\bm{\alpha}}$.  
A \emph{term} $t$ is a product $t=c\cdot \bm{x}^{\bm{\alpha}}$ where $c$ is an element in $R$ (as the coefficient of the term) and $\bm{x}^{\bm{\alpha}}$ is a monomial. 
When ${\bm{\alpha}}$ is the zero vector, i.e., $\myangle{0,\dots, 0}$, then the term $c\cdot \bm{x}^{\bm{\alpha}}$ is treated as the constant element $c$.  
A \emph{polynomial} $f$ is a finite sum of terms, i.e., 
$f= \sum_{i=1}^\ell t_i$, where each $t_i$ is a term. 
The \emph{degree} of a polynomial $f=\sum_{i=1}^\ell t_i$ with each term $t_i=c_i\cdot \bm{x}^{{\bm{\alpha}}_i}$, denoted by $\mydeg{f}$, is defined as $\max_{1\le i\le \ell} \abs{{\bm{\alpha}}_i}$.
The set of all polynomials with variables $x_1,\ldots, x_n$ and coefficients from a ring $R$ is denoted by $R[\bm{x}]$ (or $R[x_1,\ldots, x_n]$).
Moreover, we denote the set of monomials and terms with variables $x_1,\ldots, x_n$ and coefficients from a ring $R$ by $M_R[\bm{x}]$ and $T_R[\bm{x}]$, respectively. 
It is straightforward to verify that the polynomials in $R[\bm{x}]$ with the intuitive addition and multiplication operations form a ring. See~\cite [Definition~1.1.3]{greuel2008singular} for details. 

Given a polynomial $f\in R[\bm{x}]$, the polynomial evaluation $f(\bm{v})$ at a vector $\bm{v}=\myangle{v_1,\dots, v_n}\in R^n$ is defined as the value calculated by substituting each $x_i$ with $v_i$ in the polynomial $f$. For a subset $F\subseteq R[\bm{x}]$, we define the \emph{variety} of $F$, denoted by $\mathcal{V}(F)$, as the set of common roots of polynomials in $F$. Formally, $\mathcal{V}(F):=\{\bm{v}\in R^{n}\mid f(\bm{v})=0\mbox{ for all }f\in F\}$. 

Let $d$ be a positive integer. 
An unsigned \emph{bit-vector} of size $d$ is a vector of $d$ bits and represents an integer in $[0, 2^d-1]$. 
The arithmetics of unsigned bit-vectors of size $d$ follow the modular addition and multiplication of the ring $\mathbb{Z}_{2^d}$. 
In this work, we call $\mathbb{Z}_{2^d}$ the ring of bit-vectors of size $d$ and focus on polynomials in $\mathbb{Z}_{2^d}[\bm{x}]$ where $x_1,\dots, x_n$ are variables that take integer values from $[0, 2^d-1]$. 

\subsection{Strong Gr\"{o}bner Bases}
Strong Gr\"{o}bner bases~\cite{norton2001strong} are an extension of classical Gr\"{o}bner bases~\cite{groebnerbasis} that deal with the membership of ideals generated by a set of polynomials over a principal ideal ring. 
To present strong Gr\"{o}bner bases, we recall ideals as follows.

\smallskip
\noindent{\em Ideals.} 
An \emph{ideal} $I$ of a ring $R$ is a subset $I\subseteq R$ such that (i) the algebraic structure $(I, +)$ forms a group (i.e., $I$ is closed under the addition and the additive inverse), and (ii) for all $a\in I$ and $r\in R$, we have that $ra\in I$. 
For a finite subset $S=\{a_1, \ldots, a_n\}$ of the ring $R$, we define the set  $\myangle{S} := \left\{r_1\cdot a_1 + \ldots + r_m\cdot a_m\ :\ r_i\in R\right\}$  as the ideal \emph{generated} by $S$.
Conversely, $S$ is called a \emph{generating set} of the ideal $\left\langle S\right\rangle$.
When $S$ is a subset of a subring $R'$ of $R$, and $r_i$ comes from $R$, we especially denote the above set by $\myangle{S: R'}_{R}$.
We abbreviate  $\myangle{S: R'}_{R}$ as  $\myangle{S}_{R}$.
By the definition of variety, it can be verified that $\mathcal{V}(S) = \mathcal{V}(\myangle{S})$ when $S\subseteq R[\bm{x}]$.
An ideal $I\subseteq R$ is \emph{principal} if $I=\myangle{a}$ for some $a\in R$. 
A fundamental problem is the ideal membership that asks to check whether an element $b\in R$ belongs to the ideal $\myangle{S}$ generated by a finite set $S\subseteq R$. 

Below we fix a ring $R$ and $n$ variables $x_1,\ldots,x_n$. 
The key ingredient in strong Gr\"{o}bner bases is a polynomial reduction operation called \emph{strong reduction}. To present strong reduction, we first introduce monomial orderings as follows.  

\smallskip
\noindent\emph{Monomial orderings}. 
A \emph{monomial ordering} $\prec$ over $M_R[\bm{x}]$ is a total ordering over $M_R[\bm{x}]$.
In this work, we consider \emph{well-ordered} monomial ordering, which means $\bm{x}^{\textbf{0}} \prec \bm{x}^{\bm{\alpha}}$ for any $\bm{\alpha}\neq 0$ and for any monomials $p, q, r$, if $p \prec q$, then $p\cdot r \prec q\cdot r$.
For example, the \emph{lexicographical ordering} (lex) orders the monomials lexicographically by $\myangle{\alpha_1, \ldots, \alpha_n}$. 
The \emph{graded-reverse lexicographical ordering} (grevlex)  is defined as the lexicographical ordering over the tuple $\myangle{\alpha_1 + \ldots + \alpha_n, \alpha_1, \ldots, \alpha_n}$.

Given a polynomial $f\in R[\bm{x}]$ and a well-ordered monomial ordering $\prec$, the \emph{leading monomial} of $f$ is the maximum monomial among all the monomials in the finite sum of terms of the polynomial $f$ under the ordering $\prec$.
We denote the leading monomial of a polynomial $f$ by $\lm(f)$.
For example, if $x_2 \prec x_1$ and $f= x_1x_2 -2x_1^2x_2$, the leading monomial of $f$ under grevlex (i.e., $\lm(f)$) is $x_1^2x_2$.
Additionally, the \emph{leading term} $\lt(f)$ and resp. \emph{leading coefficient} $\lc(f)$ of a polynomial $f$ are defined as the term containing the leading monomial and resp. the coefficient of the leading term.
Hence, $\lt(f) = -2x_1^2x_2$ and $\lc(f) = -2$.

\smallskip
\noindent\emph{Strong Gr\"obner bases}. 
Strong Gr\"{o}bner bases provide a way to check the membership of an ideal of a polynomial ring whose coefficients are from a principal ideal ring. 
Here we first show the definition of these rings.
\begin{definition}[Principal Ideal Ring \cite{greuel2008singular}] A \emph{principal ideal ring} (PIR) is a ring $R$ such that every ideal $I\subseteq R$ is principal. 
\end{definition}

Especially, for any prime number $p$, the ring $\mathbb{Z}_{p^k}$ is a PIR.
Then, we present the strong reduction used in strong Gr\"{o}bner bases as follows. 

\begin{definition}[Strong Reduction]
Let $G\neq \emptyset$ be a finite subset of $R[\bm{x}]\setminus \{0\}$ and $f,h\in R[\bm{x}]\setminus \{0\}$.
Then we say that $f$ \emph{strongly reduces} to $h$ with respect to $G$, 
denoted by $f\twoheadrightarrow_G h$, if there exists $t\in T_R[\bm{x}], g \in G$ such that $h = f-t\cdot g$ and $\lt(f) = t\cdot\lt(g)$.
Denote the transitive closure of relation $\twoheadrightarrow_G$ by $\twoheadrightarrow_G^*$.
\end{definition}

Note that when there exists $h$ such that $f\twoheadrightarrow_G h$, we say $f$ is \emph{strongly reducible with respect to $G$}.
Otherwise, we say $f$ is \emph{irreducible with respect to $G$}.
Now, we present the definition of strong Gr\"{o}bner bases, which work for PIRs.

\begin{definition}[Strong Gr\"{o}bner Basis \cite{norton2001strong}]
A \emph{strong Gr\"{o}bner basis} for an ideal $I\subseteq R[\bm{x}]$ is a finite subset $G\subseteq I$ such that for any polynomial $f\in R[\bm{x}]$, $f\in I$ if and only if $f\twoheadrightarrow_G^* 0$.
\end{definition}

While strong Gr\"{o}bner bases do not always exist for general rings, they exist for PIR.
When $R$ is a PIR, Norton et al.~\cite{norton2001strong} propose a constructive algorithm for computing them in finite time, which can be summarized as follows.
\begin{theorem}
When $R$ is a PIR and $F\subseteq R[\bm{x}]$ is finite, Algorithm 6.4 of~\cite{norton2001strong} always returns a strong Gr\"{o}bner basis of $\myangle{F}$ in finite time.
\end{theorem}

\section{Quantifier-Free Equational Bit-Vector Theory}
\label{sec:smt_acceleration}

In this section, we introduce a novel approach for solving SMT formulas in the quantifier-free equational bit-vector theory. This section is organized as follows. First, we present the overall framework for SMT solving via strong Gr\"{o}bner bases. Second, we propose a key algorithmic improvement in the computation for strong Gr\"{o}bner bases, namely the calculation of the multiplicative inverse in the ring  $\mathbb{Z}_{2^d}$ of bit-vectors of size $d$. 

Below we fix the size $d$ for bit-vectors and $n$ variables $x_1,\ldots, x_n$ where each variable takes values from the ring $\mathbb{Z}_{2^d}$. 
Let $V=\{x_1,\ldots, x_n\}$.
Formulas in the quantifier-free equational bit-vector theory are given by the following grammar:
\[
\phi ::= f = g\,\mid\, f\ne g\, \mid f = ite(\phi, g, h) \,\mid\, \phi\vee\phi\, \mid\, \phi \wedge\phi\,\mid \neg\phi
\]
where $f,g, h\in \mathbb{Z}_{2^d}[\bm{x}]$. Informally, the quantifier-free equational bit-vector theory covers boolean combinations of (in)equations between polynomials in $\mathbb{Z}_{2^d}[\bm{x}]$. 

Note that in our grammar there is no distinction between signed and unsigned bit-vectors since we only consider equalities (i.e., $f=g$) and strict inequalities (i.e., $f \neq g$) and there are no bit-wise operations (like bit-or, bit-and, etc).

\subsection{SMT Solving with Strong Gr\"{o}bner Bases}

Our approach is built upon the framework of DPLL(T)~\cite[Chapter 11]{DBLP:series/txtcs/KroeningS16} with conflict-driven clause learning (CDCL)~\cite[Chapter 4]{DBLP:series/faia/336}, which is a technique widely adopted in modern SMT solvers. 
The general workflow of DPLL(T) is as follows. DPLL(T) regards each atomic predicate in the SMT formula $\phi$ as a propositional variable so that the original SMT formula is transformed into a propositional formula and partially assigns truth values to these propositional variables in a back-tracking procedure while checking whether conflict arises. When a conflict is detected, DPLL(T) tries to learn a new clause through CDCL to guide the subsequent SMT solving. 

A central component of a DPLL(T) solver is the satisfiability checking of a conjunction of atomic predicates and their negations. 
In the DPLL(T) solving of our quantifier-free equational bit-vector theory, the atomic predicates are equational predicates of the forms $f=g$ (with their negations $f\ne g$). Therefore, the aforementioned central component corresponds to the satisfiability checking of a system of polynomial (in)equations modulo $2^d$ for bit-vectors. 

We solve the satisfiability of a system of polynomial (in)equations modulo $2^d$ via strong Gr\"{o}bner bases. The detailed algorithmic steps are as follows. Below we fix an input finite conjunction $\Phi=\bigwedge_i \phi_i$ where each $\phi_i$ is either an equation $f=g$ or an inequality $f\ne g$ where $f,g\in \mathbb{Z}_{2^d}[\bm{x}]$.  

\smallskip
\noindent{\em $\blacktriangleright$ Step A1: Pre-processing.} Our algorithm first operates a pre-processing that transforms each (in)equation $\phi_i$ in the conjunction $\Phi$ into an equivalent equation of the form $h=0$ for some polynomial $h$ in $\mathbb{Z}_{2^d}[\bm{x}]$. 
If $\phi_i$ is an equation $f_i=g_i$, then we simply set $h_i=f_i-g_i$.
Otherwise, if $\phi_i$ is an inequation $f_i\ne g_i$, let $h_i=z_i (f_i-g_i)-2^{d-1}$ where $z_i$ is a fresh variable.

The idea behind introducing the fresh variable $z_i$ to handle inequations is that in the ring $\mathbb{Z}_{2^d}$, we have $a\not\equiv 0\mod 2^d$ if and only if $a\cdot b \equiv 2^{d-1}\mod 2^d$ for some $b\in \mathbb{Z}_{2^d}$. 
We then denote the collection of all the $h_i$ by $H$.
The following propositions demonstrate the correctness of the pre-processing. 
The proofs of Proposition~\ref{prop:odd_inverse} and Proposition~\ref{prop:equiv_cond_for_phi_sat} are deferred to Appendix~\ref{appendix:omitted_proof_of_prop1} and Appendix~\ref{appendix:omitted_proof_of_prop2}.

\begin{proposition}\label{prop:odd_inverse}
Given $a\in \mathbb{Z}_{2^d}\setminus \{0\}$ with $\nu_2(a) = \alpha$, then there exists an integer $b\in  \mathbb{Z}_{2^d}$ such that $ab= 2^\alpha$.
Especially, if $\alpha=0$, $b$ is unique.
\end{proposition}

\begin{proposition}\label{prop:equiv_cond_for_phi_sat}
$\Phi$ is satisfiable if and only if $\mathcal{V}(H)$ is non-empty.
\end{proposition}

After the preprocessing, our algorithm constructs a finite set of polynomials $H$.
According to the above proposition, to check whether $\Phi$ is satisfiable, it suffices to examine the emptiness of $\mathcal{V}(H)$. 
We present this in the next step.

\smallskip
\noindent{\em $\blacktriangleright$ Step A2: Emptiness checking of $\mathcal{V}(H)$.}
\label{sec:compute_groebner_bases} In the previous step, we 
have reduced the satisfiability of the original SMT formula $\Phi$ into the existence of a common root of the polynomials in $H$ modulo $2^d$.
Our strategy is first to try witnessing the emptiness of $\mathcal{V}(H)$ via strong Gr\"{o}bner bases for the ideal $\myangle{H}$, and then resort to the root-finding of polynomials if the witnessing fails. 
The witnessing through strong Gr\"{o}bner bases has the potential to determine the unsatisfiability of $\Phi$ quickly and, therefore, can speed up the overall SMT solving. 

By the definition of ideals, one has that $\mathcal{V}(H)$ is empty if there is some nonzero element $c\in \mathbb{Z}_{2^d}$ in $\myangle{H}$. 
This is given by the following proposition. 

\begin{proposition}\label{prop:empvariety}
If there exists a nonzero $c\in \mathbb{Z}_{2^d}$ in $\myangle{H}$, then 
$\mathcal{V}(H)=\emptyset$. 
\end{proposition}
The above proposition suggests that to witness the emptiness of $\mathcal{V}(H)$, it suffices to find a nonzero constant $c\in \mathbb{Z}_{2^d}$ in the ideal $\myangle{H}$. The following theorem establishes a connection between the existence of such constant in the ideal $\myangle{H}$ and the strong Gr\"{o}bner basis for the ideal. 

\begin{theorem}\label{thm:const}
There exists a constant nonzero polynomial $f\in \mathbb{Z}_{2^d}$ in $\myangle{H}$ if and only if for any strong Gr\"{o}bner basis $G$ for the ideal $\myangle{H}$, $G$ contains some nonzero constant $g \in \mathbb{Z}_{2^d}$. 
\end{theorem}
\begin{proof}
According to the definition of the strong Gröbner bases, for any polynomial $f$, $f \in \myangle{H}$ if and only if $f \twoheadrightarrow_{G}^* 0$.
It implies there exists a polynomial $g\in G$, a monomial $m\in M_R[\bm{x}]$ and a coefficient $c\in \mathbb{Z}_{2^n}$ such that $\lt(f) = c m \cdot \lt(g)$. 
If $\deg(f)= 0$, then $\lt(f)$ is $f$ itself.
In addition, since $\bm{x}^{\textbf{0}} = \lm(f) = m\cdot \lm(g)$, we have $m = \lm(g) = \bm{x}^{\textbf{0}}$.
Thus, according to the property of well-ordered ordering, $g$ is also a constant polynomial.
Since $f\neq 0$, $g$ is also a nonzero constant polynomial.
Conversely, if there is a nonzero constant within $G$, it certainly implies that there exists a nonzero constant in $\myangle{H}$.
\end{proof}

By Theorem~\ref{thm:const}, we reduce the witness of a nonzero constant polynomial in the ideal to that in a strong Gr\"{o}bner basis. 
Thus, our algorithm computes a strong Gr\"{o}bner basis for the ideal $\myangle{H}$ to witness the emptiness of $\mathcal{V}(H)$ and the unsatisfiability of the original conjunction $\Phi$. 
If the witnessing fails, our algorithm resorts to existing computer algebra methods to check the emptiness of $\mathcal{V}(H)$. 

In Algorithm~\ref{alg:compute_groebner_bases}, we implement the original algorithm in~\cite{norton2001strong} to compute a strong Gr\"{o}bner basis $GB(H)$ of the ideal $\myangle{H}$. 
The basic idea is to iteratively construct new polynomials from existing ones $f$ derived from $H$ and add them to the candidate bases $G$ until convergence.
Before describing this algorithm, we first introduce two key concepts used to construct new polynomials, namely \emph{S-polynomials} and \emph{A-polynomials}, which are defined as follows.
Both provide approaches to constructing new polynomials from original polynomials of $G$.

\begin{definition}[S-polynomials~\cite{norton2001strong}]
Given two distinct polynomials $f_1, f_2\in R[\bm{x}]\setminus \{0\}$, a polynomial in form of  
$c_1m_1f_1 - c_2m_2 f_2$ where $c_1, c_2 \in R$ is called a \emph{S-polynomial}, if
$c_1 \lc(f_1) = c_2\lc(f_2)$ is a least common multiple of $\lc(f_1)$ and $\lc(f_2)$, and $m_i = \mathrm{lcm}(\lm(f_1), \lm(f_2))/\lm(f_i)\in M_R[\bm{x}]$, where $ \mathrm{lcm}(\lm(f_1), \lm(f_2))$ is the least multiple of $\lm(f_1)$ and $\lm(f_2)$.
We denote the set of S-polynomials of $f_1$ and $f_2$ by $\mathrm{Spoly}(f_1, f_2)$.
\end{definition}

\begin{definition}[A-polynomials~\cite{norton2001strong}]
Given $f\in R[\bm{x}]\setminus \{0\}$, an \emph{A-polynomial} of $f$ is any polynomial in form of $a\cdot f$, where $\myangle{a}_{R} = \mathrm{Ann}(\lc(f))$ and $\mathrm{Ann}(x) := \{a:ax = 0\}$.
We denote the set of A-polynomials of $f$ by $\mathrm{Apoly}(f)$.
\end{definition} 
Besides, a new polynomial can also be generated by iteratively reducing $f$ with respect to the polynomials in $G$ until a remainder 
$h$ is obtained, which is irreducible.
The remainder $h$ is defined as the \emph{normal form of $f$ with respect to $G$} and denoted by $\NF(f\ |\ G)$.
It can be verified that $f\twoheadrightarrow_G^* \NF(f\ |\ G)$, and if $G$ is a Gr\"obner basis of $\myangle{H}$, $f\in \myangle{H}$ if and only $\NF(f\ |\ G) = 0$.

\begin{algorithm}
	\caption{Finding a strong Gr\"obner basis of $\myangle{H}$}
	\label{alg:compute_groebner_bases}
	\begin{minipage}{0.5\textwidth}
		\begin{algorithmic}[1] 
			\State $G\leftarrow H$ and $C\leftarrow H$\;
                \State $P\leftarrow \{(f_1, f_2)\ :\ f_1\neq f_2\in H\}$\;
		    \While{$P\neq \varnothing$ or $C\neq \varnothing$} 
                    \State $h\leftarrow 0$\;
                    \If{$C\neq \varnothing$}
                        \State Pop a polynomial $f_1$ from $C$\;
                        \State $h\leftarrow \apoly{f_1}$\; \label{line:compute_apoly}
                    \Else 
                        \State Pop a pair $(f_1, f_2)$ from $P$\;
                        \State $h \leftarrow \spoly{f_1, f_2}$\;\label{line:compute_spoly}
                    \EndIf
                    \State Compute $g\leftarrow \NF(h\ |\ G)$\; \label{line:cal_normal_form}
                    \If{$g\neq 0$}
                        \State $P\leftarrow P\cup \{(g, f):f\in G\}$\; \label{line:update_B}
                        \State $C \leftarrow C\cup \{g\}, G \leftarrow G\cup \{g\}$\;\label{line:update_C}
                    \EndIf
                \EndWhile
                \Return{$G$}
            \end{algorithmic}
	\end{minipage}
	\begin{minipage}{0.5\textwidth}
		\begin{algorithmic}[1]
            \renewcommand{\algorithmicrequire}{\textbf{Parameters:}}	
            \Require $f, G$
            \Ensure $\NF(f\ |\ G)$
            \Function{$\NF(f\ |\ G)$}{}
                \While{there exists a $g\in G$ and a term $t=c_t\cdot m_t$ of $f$, s.t., $\nu_2(\lc(g)) \le \nu_2(c_t)$ and $\lm(g)\ |\ m_t$}
                    \State $f \leftarrow f -  \frac{t}{\lt(g)} \cdot g$\; \label{line:division_of_normal_form} 
                \EndWhile
                \Return{$f$}
            \EndFunction
		\end{algorithmic}
	\end{minipage}
\end{algorithm}

Below, we are ready to give a brief description of Algorithm~\ref{alg:compute_groebner_bases}. 
Given a set of polynomials $H$ as input, it initializes the candidate polynomials as $G = H$. 
Meanwhile, it sets $C = H$ as the set of polynomials required to compute their A-polynomials and $P$ as the set of pairs required to compute S-polynomials.
When the set $C$ is non-empty, it pops a polynomial $f\in C$ and computes one of its A-polynomials $h$ (Line~\ref{line:compute_apoly}).
If its normal form $g= \mathrm{NF}(h\ |\ G)\neq 0$, it appends $g$ into $G$ of candidate Gr\"obner basis.
Besides, it updates $P$ and $C$ through building new pairs $\{(g, f):g\in G\}$ and appending $g$ to $C$ (Line~\ref{line:update_B},\ref{line:update_C}).
When $C$ is empty, it pops a pair $(f_1, f_2)$ from $P$ and computes its S-polynomial $h$ (Line~\ref{line:compute_spoly}).
When the normal form of $h$ is non-zero, we do the same operation as above.
In particular, we instantiate the A-polynomial and S-polynomial selected on Line~\ref{line:compute_apoly} and Line~\ref{line:compute_spoly} by $\apoly{f_1}$ and $\spoly{f_1,f_2}$, whose definitions are as follows.
\begin{align*}
& \textsf{apoly}(f_1) := 2^{d-k_1} \cdot f_1, \\ 
& \textsf{spoly}(f_1, f_2)  := 2^{d-k_1}\cdot s_2\cdot \bm{x}^{\bm{\alpha}-\bm{\alpha_1}} \cdot f_1 - 2^{d-k_2}\cdot s_1\cdot \bm{x}^{\bm{\alpha}-\bm{\alpha_2}} \cdot f_2,
\end{align*}
where $\bm{x}^{\alpha_i} = \lm(f_i), \bm{x}^{\alpha} = \mathrm{lcm}(\bm{x}^{\bm{\alpha_1}}, \bm{x}^{\bm{\alpha_2}}), k_i = \nu_2(\lc(f_i))$ and $\lc(f_i) = 2^{k_i}\cdot s_i$. 
The correctness of Algorithm~\ref{alg:compute_groebner_bases} is given in Theorem~\ref{prop:groeb_cond_over_gr}, with proof in Appendix~\ref{appendix:omitted_proof_of_prop:groeb_cond_over_gr}.

\begin{algorithm}[t]
\caption{Finding common roots of a Gr\"obner basis.}
\label{alg:find_zeros}
\begin{algorithmic}[1]
\Function{$FindZeros(H, \mathcal{M})$}{}
\State $G \leftarrow GB(H)$\;
\State \textbf{if} \textit{there exists some constant in $G$} \textbf{then} \textbf{return} $\perp$
\State \textbf{if} $\abs{\mathcal{M}} = n$ \textbf{then} \textbf{return} $\mathcal{M}$\;  
\If{there exists an univariate polynomial $p\in \mathbb{Z}_{2^n}[x]$}
\State \textbf{if} \textit{$Zeros(p) =\emptyset$}  \textbf{then} \textbf{return} $\perp$\; 
\For{$z\in Zeros(p)$} \Comment{Traverse zeros of a univariate polynomial}
\State $r\leftarrow FindZeros(G\cup \{x - z\}, \mathcal{M} \cup \{x \mapsto z\})$\; \label{line:solve_univar}
\State \textbf{if} \textit{$r\neq \perp$}  \textbf{then} \textbf{return} $r$\; 
\EndFor
\ElsIf{there exists a polynomial $p$ can be decomposed as $p = f\cdot g$}
\For{$i= 0, \ldots, d-1$} \Comment{Factorize $p$ over $\mathbb{Z}$}
\State $r\leftarrow FindZeros(G\cup \{2^i\cdot f, 2^{d-i}\cdot g\}\setminus \{p\}, \mathcal{M})$\;
\State \textbf{if} \textit{$r\neq \perp$}  \textbf{then} \textbf{return} $r$\; \label{line:factorization}
\EndFor
\Else
\State Arbitrarily select a variable $x \notin \mathcal{M}$\; \Comment{Exhaustive search}
\For{$z= 0, \ldots, 2^d-1$}
\State $r\leftarrow FindZeros(G\cup \{x - z\}, \mathcal{M}\cup \{x \mapsto z\})$\;
\State \textbf{if} \textit{$r\neq \perp$}  \textbf{then} \textbf{return} $r$\; \label{line:exhaust}
\EndFor
\EndIf
\State \Return{$\perp$}\;
\EndFunction
\end{algorithmic}
\end{algorithm}

\begin{theorem}
\label{prop:groeb_cond_over_gr}
For any finite set $H\subseteq \mathbb{Z}_{2^d}[x_1,\ldots, x_n]$, Algorithm~\ref{alg:compute_groebner_bases} always returns a strong Gr\"obner basis of $\myangle{H}$ in finite time.
\end{theorem}

\smallskip
\noindent{\em $\blacktriangleright$ Step A3: Guarantee on completeness.}
As previously discussed, we have presented a sound procedure for determining whether $\mathcal{V}(H)$ is empty.
However, there may be cases where $\mathcal{V}(H)$ is empty despite passing the aforementioned check.
To address this issue, we have developed a complete procedure that utilizes a backtracking strategy to identify solutions.

Similar to the complete decision procedure of~\cite{DBLP:conf/cav/OzdemirKTB23}, we maintain two data structures: $G$, a Gr\"obner basis and $\mathcal{M}: V\rightarrow \mathbb{Z}_{2^n}$, a (partial) map from variables to elements of ring $\mathbb{Z}_{2^n}$.
Algorithm~\ref{alg:find_zeros} presents the \emph{FindZeros} procedure for finding solutions.
$G$ is first initialised to $GB(H)$ and $\mathcal{M}$ is initialised to an empty map.

We say a polynomial $p$ is \emph{unvariate} if $p$ only contains one variable not assigned a value in $\mathcal{M}$.
If $G$ contains a univariate polynomial $p$ with only one variable $x$ that is not assigned a value, we compute the zeros of $p$ over $\mathbb{Z}_{2^n}$ (denoted by $Zeros(p)$). 
Then, we traverse each root $z\in Zeros(p)$ and iteratively assign $x$ as $z$.
For each assignment, we recursively check whether the updated Gr\"obner basis has a solution.
If each polynomial has more than two variables not assigned, we attempt to decompose a polynomial $p\in G$ into $p = f\cdot g$ through factorization over $\mathbb{Z}$. 
Since $p = 0\ (\mathrm{mod}\ 2^d)$ is equivalent to $f =0\ (\mathrm{mod}\ 2^{d-i})$ and $g= 0\ (\mathrm{mod}\ 2^i)$ for some $0\le i \le d$, we enumerate $i$ and recursively search for solutions by appending $f\cdot 2^i$ and $g\cdot 2^{d-i}$ into the Gr\"obner basis.
When neither of the two conditions mentioned above are met, we then perform an exhaustive search until a solution is found.
The soundness of the complete decision procedure is given by Theorem~\ref{thm:correctness_of_find_zeros}, whose proof is deferred to Appendix~\ref{appendix:omitted_proof_of_correctness_of_find_zeros}.
\begin{theorem}\label{thm:correctness_of_find_zeros}
Algorithm~\ref{alg:find_zeros} always terminates and returns an element of $\mathcal{V}(H)$ if $\mathcal{V}(H) \neq \emptyset$.
Otherwise, it returns $\perp$.
\end{theorem}

\subsection{Algorithmic Improvement for Multiplicative Inverse}
In the computation of a strong Gr\"{o}bner basis, a key bottleneck is the division operation (Line~\ref{line:division_of_normal_form} of NF in Algorithm~\ref{alg:compute_groebner_bases}).
Given two integers $a, b\in \mathbb{Z}_{2^n}$, let $a = 2^{\nu_2(a)}\cdot s_a$ and $b= 2^{\nu_2(b)}\cdot s_b$, where $s_a,s_b$ are both odd integers.
According to Proposition~\ref{prop:odd_inverse}, there exists an unique integer $t$ such that $t\cdot s_b = 1\ (\mathrm{mod}\ 2^d)$.
We denote $t$ by $s_b^{-1}$ and implement the division operation  as $\frac{a}{b} = 2^{\nu_2(a) - \nu_2(b)}\cdot s_a \cdot s_b^{-1}$.

There are various ways to calculate the inverse of an odd integer $a$ in $\mathbb{Z}_{2^d}$.
A simple method is via the classical extended Euclidean algorithm~\cite[Chapter 2]{elemnumber} that finds integer coefficients $k_1,k_2$ such that $k_1\cdot a+k_2\cdot {2^d}=1$. 
The extended Euclidean algorithm requires $O(d)$ arithmetic operations. 
A significant improvement is via Hensel’s lifting~\cite{krishnamurthy1983fast} that requires only $O(\log d)$ arithmetic operations.
Let $a = 2^s\cdot a' + 1$, where $a'$ is an odd integer.
Prominent implementations of Hensel’s lifting include Arazi and Qi's algorithm~\cite{arazi2008calculating} (that requires $13\lfloor \log d \rfloor + 1$ arithmetic operations) and Dumma's algorithm~\cite{dumas2013newton} (that requires $5\lfloor \log (\frac{d}{s}) \rfloor + 2$ arithmetic operations).
They adopt constructive methods through recursive formulas.
In particular, Dumma's algorithm calculates the inverse of $a$ through iteratively calculating $(2-a)\prod_{i=1}^{n-1}(1+ (a-1)^{2^i})$.

A limitation of methods using Hensel's lifting is constructing the inverse from scratch, even if $a$ is small.
For example, when $a=3$, Dumma's algorithm will iteratively calculate $-\prod_{i=1}^{n-1}(1+ 2^{2^i})$ while $3^{-1}$ could be simply given by $\frac{2^{d+1} +1}{3}$.
The cost of construction will be non-negligible when $2^d$ is extremely large.

\begin{wrapfigure}[9]{L}{0.5\textwidth} 
\begin{minipage}{0.5\textwidth}
\vspace{-7mm}
\begin{algorithm}[H]
\caption{Finding the multiplicative inverse of $a$ over $\mathbb{Z}_{2^d}$ when $a$ is small.}
\label{alg:find_inverse}
\begin{algorithmic}[1]
\State $r\leftarrow 2^d \mod a$\;\label{line:2_d_mod_a}
\State Implement $f\leftarrow \frac{2^d - r}{a}$ by $f\leftarrow (\frac1a)\gg d$ \label{line:div_operation}
\State Apply extended Euclidean algorithm to find $k_1, k_2$, such that $k_1\cdot r + k_2\cdot a = -1$\label{line:extended_euclidean}
\State \Return{$(k_1\cdot f- k_2)\mod 2^d$} \label{line:return_final_res}
\end{algorithmic}
\end{algorithm}
\end{minipage}
\end{wrapfigure}

To improve this, our observation is that the arithmetic operations over $\mathbb{Z}_a$ will be cheaper when $a$ is small compared to $2^d$.
Since finding the multiplicative inverse of $a$ is equivalent to finding $k$ such that $\frac{k\cdot 2^d + 1}{a}$ is an integer, it could be achieved by finding the inverse of $2^d$ over $\mathbb{Z}_a$.

Based on this insight, we adopt Algorithm~\ref{alg:find_inverse} to find the inverse of $a$ when $a$ is small.
Its detailed description is deferred to Appendix~\ref{appendix:proof_of_complexity_of_mul_inverse}.
Here we present a high-level description.
First, it computes the remainder of $2^d$ modulo $a$, which can be implemented by the modular exponentiation algorithm~\cite{schneier2007applied}.
Then, we use Newton's method and a shifting operation to calculate the result of $a$ dividing $2^d-r$.
Next, we invoke the extended Euclidean algorithm to find integer $k_1, k_2$ such that $k_1\cdot r+ k_2\cdot a = -1$.
Finally, we return $k_1\cdot f - k_2$ as the inverse of $a$, whose correctness is given by
$$
a\cdot (k_1\cdot f - k_2) = k_1 \left(2^d- r\right) - k_2\cdot a \equiv 1\ (\mathrm{mod}\ 2^d)\,.
$$

It can observed that most arithmetic operations (Line~\ref{line:2_d_mod_a} and~\ref{line:extended_euclidean}) in Algorithm~\ref{alg:find_inverse} are over $\mathbb{Z}_a$. 
Hence, when $a$ is small compared to $2^d$, the number of binary operations over $\mathbb{Z}_a$ will also be small compared to that over $\mathbb{Z}_{2^d}$.
Assume the classical multiplication costs $\mathcal{O}(2d^2)$ binary operations.
A division operation $\frac{b}{a}$ could be implemented by pre-computing $\frac1a$ and then performing one multiplication $b\cdot (\frac{1}{a})$ as~\cite{granlund1994division}.
Formally, we provide the following lemma, which estimates the number of arithmetic and binary operations.
Notably, our method outperforms those of~\cite{arazi2008calculating,dumas2013newton} when $a$ is much smaller than $2^d$.
Its proof is deferred to Appendix~\ref{appendix:proof_of_complexity_of_mul_inverse}.
\begin{theorem}\label{thm:complexity_of_mul_inverse}
Algorithm~\ref{alg:find_inverse} requires $\mathcal{O}\left(2\log d + 8\log a \right)$ arithmetic operations and $\mathcal{O}(2d^2 + 4d + 8 a^2\cdot \log a + 4a\cdot \log a)$ binary operations.
\end{theorem}

\section{Loop Invariant Generation}
\label{sec:loop_invariant}

In this section, we introduce a novel approach for generating polynomial equational invariants over bit-vectors.
Invariant generation aims at solving predicates that over-approximate the set of reachable program states and can be viewed as an important case of quantified SMT solving~\cite{DBLP:conf/kbse/YaoKSFWR23}. 
We solve the invariant by parametrically reducing a polynomial invariant template with unknown coefficients with respect to the strong Gr\"{o}bner basis of transitions.
Below we first present the invariant generation problem and the connection with SMT. Then we present the details of our approach. 

\subsection{Inductive Loop Invariants}
\label{sec:intro_to_linear_inv}

We consider polynomial while loops in which addition and multiplication are considered modulo $2^d$ where $d$ is the size of a bit-vector.
Fix a finite set $V=\{x_1,\dots, x_n\}$ of program variables and let $V=\{x'\mid x\in V\}$ be the set of primed variables of $V$. A program variable $x$ in $V$ is used to represent the current value of the variable, while its primed counterpart $x'\in V'$ is to represent the next value of the variable $x$ after the current loop iteration. An \emph{equational polynomial assertion} is a conjunction of polynomial equations over variables in $V, V'$, i.e., a formula of the form $\bigwedge_{i} p_i(x_1,x_1',\dots, x_n,x'_n)= 0$ where each $p_i$ is a polynomial with variables from $V,V'$. 
More generally, a \emph{polynomial assertion} includes polynomial inequalities, i.e.,  $\bigwedge_{i} p_i(x_1,x_1',\dots, x_n,x'_n) \Join q_i(x_1,x_1',\dots, x_n,x'_n)$, where $\Join \in \{=,\neq, \le, \ge, >, < \}$ and $p_i, q_i$ are polynomials with variables from $V,V'$.

We say that a polynomial assertion is in $V$ if it only involves variables from $V$ and in $V, V'$ generally. 
A polynomial while loop takes the form
\begin{equation}\label{eq:polynomialloop}
\textbf{assume}\  \theta(V);\,\textbf{while}\,(c(V))\,\{\rho(V, V');\}; \textbf{assert}(\kappa(V)) 
\end{equation}
where (a) $\theta(V)$ is a polynomial assertion in $V$ and specifies the initial condition (or precondition) for program inputs, (b) $c(V)$ is a polynomial assertion in $V$ and acts as the loop condition, (c) $\rho(V, V')$ is an equational polynomial assertion in $V, V'$ that specifies the relationship between the current values of $V$ and the next values of $V'$, and (d) $\kappa(V)$ is a polynomial assertion in $V$ that acts as the postcondition at the termination of the loop. 

A \emph{program state} is a function $\bm{v}:V\rightarrow \mathbb{Z}_{2^d}$ that specifies the current value $\bm{v}(x)$ for every variable $x\in V$, while a \emph{primed state} is a function $\bm{v}':V'\rightarrow  \mathbb{Z}_{2^d}$ that specifies the next value $\bm{v}'(x')$ of each variable $x$ after one loop iteration. 
Given a program state $\bm{v}$ and a polynomial assertion $\phi$ in $V$, we write $\bm{v} \models \phi$ if $\phi$ is true when each variable $x\in V$ is instantiated by its current value $\bm{v}(x)$ in $\phi$. Moreover, given a program state $\bm{v}$, a primed state $\bm{v}'$ and a polynomial assertion $\phi$ in $V,V'$, we write $\bm{v}, \bm{v}' \models \phi$ if $\phi$ is true when instantiating each variable $x\in V$ with the value $\bm{v}(x)$ and each variable $x'\in V$ with the value $\bm{v}'(x')$ in $\phi$.  

The semantics of a polynomial while loop in the form of (\ref{eq:polynomialloop}) is given by its (finite) executions. A (finite) \emph{execution} of the loop is a finite sequence $\bm{v}_0,\bm{v}_1,\dots, \bm{v}_n$ of program states such that $\bm{v}_0\models \theta(V)$ and for every $0\le k< n$ we have $\bm{v}_k,\bm{v}'_{k+1}\models\rho(V,V')$, where the primed state $\bm{v}'_{k+1}$ is defined by $\bm{v}'_{k+1}(x'):=\bm{v}_{k+1}(x)$ for each $x\in V$. A program state $\bm{v}$ is \emph{reachable} if it appears in some execution of the loop. We consider polynomial invariants, which are equational polynomial assertions $\phi$ in $V$ such that for all reachable program states $\bm{v}$, we have $\bm{v}\models \phi$. 

In this work, we consider inductive polynomial invariants that are inductive invariants in the form of polynomial equations. An \emph{inductive polynomial invariant} for a polynomial while loop in the form of (\ref{eq:polynomialloop})  is an equational polynomial assertion $\phi$ in $V$ that satisfies the initiation and consecution conditions as follows:
\begin{itemize}
    \item (\textbf{Initiation}) 
    for any program state $\bm{v}$, $\bm{v} \models \theta(V)$ implies $\bm{v} \models \phi$. 
    \item (\textbf{Consecution}) for any program state $\bm{v}$ and primed state $\bm{v}'$, we have that  
    \[
    \left[(\bm{v} \models c(V) \wedge \phi)\wedge (\bm{v}, \bm{v}' \models \rho(V, V'))\right] \Rightarrow\bm{v}'\models \phi[x'_1/x_1,\dots, x'_n/x_n].
    \]
\end{itemize}
By a straightforward induction on the length of an execution, one easily observes that every inductive polynomial invariant is an invariant that over-approximates the set of reachable program states of a polynomial while loop. An inductive polynomial invariant $\phi$ \emph{verifies} the post condition $\kappa(V)$ if for all program states $\bm{v}\models \phi$, we have that $\bm{v}\models \kappa(V)$. Moreover, the invariant $\phi$ \emph{refutes} the post condition if $\phi\wedge\kappa(V)$ is unsatisfiable.  

We consider the automated synthesis of inductive polynomial invariants: given an input polynomial while loop in the form of (\ref{eq:polynomialloop}), generate an inductive polynomial invariant for the loop that suffices to verify the postcondition of the loop. 
The problem is a special case of SMT solving in the CHC (constraint Horn clauses) form~\cite{DBLP:conf/sas/BjornerMR13}. The corresponding CHC constraints are as follows:
\begin{align}\label{eq:chc_constraint}
\mbox{\textbf{Initiation:}}\, &\, \theta(V)\Rightarrow \phi \nonumber \\ 
\mbox{\textbf{Consecution:}}\, &\,  \left[(c(V) \wedge \phi)\wedge \rho(V, V')\right] \Rightarrow \phi[x'_1/x_1,\dots, x'_n/x_n] \nonumber \\
\mbox{\textbf{Verification:}}\, &\,  
(\phi \wedge \neg c(V)) \Rightarrow \kappa(V)\enskip\enskip\textbf{Refutation:}\,(\phi \wedge \neg c(V)) \Rightarrow \neg \kappa(V)
\end{align}
where the task is to solve a formula $\phi$ that fulfills the constraints above, for which we use the verification condition to verify the postcondition and the refutation condition to refute the postcondition. One directly observes that any solution of $\phi$ w.r.t the CHC constraints is an inductive polynomial invariant for the loop.

\begin{example}
\label{example:smt_solving_for_inv}
Consider the following program, where $x,y$ are two $32$-bit unsigned integers and initialized as $1$ and $9$, respectively, and incremented by $1$ in each iteration.
We want to verify that $y-x < 10$ when the loop program terminates.
According to the definition of polynomial inductive loop invariant, the value of $y - x$ is fixed during the loop. 
Since its value is initialized as $8$, $y - x= 8\ (\mathrm{mod}\ 2^{32})$ is a polynomial loop invariant.
Further, since $(y = 0) \wedge (y - x = 8\ (\mathrm{mod}\ 2^{32}))$ implies $y - x<10$, the postcondition is satisfied.
\begin{center}
\begin{minipage}{0.18 \textwidth}
\begin{verbatim}
x = 1, y = 9;
while (y != 0) {
    x = x + 1;
    y = y + 1;
}
@assert(y - x < 10);
\end{verbatim}
\end{minipage}
{
\footnotesize
$\xrightarrow{\substack{\text{Embed into}\\ \text{quantified formulas}}}$
}
\begin{minipage}{0.12\textwidth}
\begin{align*}
&\textbf{assert}(\forall x, y. (x=1\wedge y=9)\rightarrow inv(x,y))\\
&\textbf{assert}(\forall x, y, x', y'. (y\neq 0 \wedge inv(x,y)) \rightarrow \\
&\qquad (x' = x+ 1 ) \wedge (y' = y + 1) \wedge inv(x', y'))\\ 
&\textbf{assert}(\forall x, y. (inv(x,y) \wedge (y = 0)) \rightarrow (y - x < 10))\\
\end{align*}
\end{minipage}
\end{center}
\end{example}

\subsection{Polynomial Invariants Generation over Bit-Vectors}
\label{sec:poly_loop_inv}

Below we present our approach for synthesizing polynomial equational loop invariants over bit-vectors. 
The high-level idea is first to have a polynomial template with unknown coefficients as parameters, then to establish linear congruence equations for these unknown coefficients via strong Gr\"{o}bner bases, and finally to solve the linear congruence equations to get polynomial invariants. 

Let $P$ be an input polynomial while loop in the form of (\ref{eq:polynomialloop}) and $d$ be the size of bit-vectors. 
Let $k$ be an extra algorithmic input that specifies the maximum degree of the target polynomial invariant and $M^{(k)}[V]$ be the set of all monomials in $V$ with degree $\le k$. Our approach is divided into three steps as follows. 

\smallskip
\noindent{\em $\blacktriangleright$ Step B1: Polynomial templates.}
Our approach first sets up a polynomial template $\eta = 0$ for the desired invariant, where
$\eta$ is the polynomial 
$
\textstyle\eta = \sum_{q\in {M^{(k)}}[V]}\lambda_q\cdot q
$
with the unknown coefficients $\lambda_q$ for every monomial $q\in M^{(k)}[V]$.
In particular, the coefficient of $q= \bm{x}^{\bm{0}}$ is also denoted as $\xi$.
For example, when $V=\{x,y\}$, and $k=2$, template $\eta$ contains monomials with degree $\le 2$, i.e.,  
$
\eta = \lambda_1\cdot x^2 + \lambda_2\cdot xy + \lambda_3 \cdot y^2 + \lambda_4 \cdot x + \lambda_5 \cdot y + \xi
$.
Denote $\bm{\lambda}$ as the vector of all these coefficients.
Note that the template enumerates all monomials with degrees no more than $k$. 
We denote $\eta':=\eta[x'_1/x_1,\dots,x'_n/x_n]$. 
We obtain the following constraints, which specify that the template is an invariant:
\begin{align}\label{eq:chc_constraint_template}
\mbox{\textbf{Initiation:}}\, &\, \theta(V)\Rightarrow \left(\eta = 0\right) \nonumber \\
\mbox{\textbf{Consecution:}}\, &\,  \left[c(V) \wedge (\eta = 0) \wedge \rho(V, V')\right] \Rightarrow (\eta'=0)  
\end{align}

\noindent
\smallskip{\em $\blacktriangleright$ Step B2: Reduction of strong Gr\"{o}bner bases.} Then, our approach derives a system of linear congruence equations for the unknown coefficients in the template with respect to the constraints in (\ref{eq:chc_constraint_template}). 
To derive a sound and complete characterization for the unknown coefficients is intractable as it requires multivariate nonlinear reasoning of integers with modular arithmetics. Henceforth, we resort to incomplete sound conditions. We derive sound conditions from the proposition below, whose correctness follows directly from the definition of ideals.

\begin{proposition}
\label{prop:sufficient_cond_for_inv}
Let $F_\rho$ be the set of left-hand-side polynomials of assertion $\rho(V, V')$. 
The consecution condition $\left[c(V) \wedge (\eta = 0) \wedge \rho(V, V') \right]\Rightarrow (\eta'=0)$ holds if $\eta' - \mu\cdot \eta \in \myangle{F_\rho}$ for some constant $\mu \in \mathbb{Z}_{2^d}$. 
\end{proposition}

Proposition~\ref{prop:sufficient_cond_for_inv} provides the sound condition $\exists\mu\in \mathbb{Z}_{2^d}.(\eta' - \mu\cdot \eta \in \myangle{F_\rho})$ for the consecution condition in (\ref{eq:chc_constraint_template}). We employ several heuristics to further simplify the condition. 
The first heuristics is to enumerate small values of $\mu$ (e.g., $\mu\in \{-1,0, 1\}$) that reflect common patterns in invariant generation. For example, the case $\mu=0$ corresponds to \emph{local} invariants, and $\mu=1$ corresponds to \emph{incremental} invariants (see \cite{DBLP:conf/popl/SankaranarayananSM04,DBLP:conf/sas/SankaranarayananSM04} for details). Then, the condition is soundly reduced to checking $\eta' - \mu\cdot \eta \in \myangle{F_\rho}$ for specific values of $\mu$. 

To check whether $\eta' - \mu\cdot \eta \in \myangle{F_\rho}$ or not, a direct method is to reduce the polynomial $\eta' - \mu\cdot \eta$ with respect to the strong Gr\"{o}bner basis of the ideal $\myangle{F_\rho}$. 
However, this would cause the combinatorial explosion as one needs to maintain the information of the unknown coefficients from the template in the reduction of strong Gr\"{o}bner bases, which is already the case in the simpler situation of real numbers (that constitute a field and do not involve modular arithmetics) \cite{DBLP:conf/popl/SankaranarayananSM04}. Therefore, we consider heuristics via a normal-form reduction as follows.

\smallskip
\noindent{\em Parametric Normal Form.}
Below we define an operation $\mathrm{PNF}$ of parametric normal form
for reducing a polynomial $f$ whose coefficients are polynomials in the unknown coefficients of our template against a finite subset of $\mathbb{Z}_{2^d}[V]$. 

\begin{definition}[Parametric Normal Form]
\label{def:param_normal_form}
Let $\mathcal{L} = \mathbb{Z}_{2^d}[\bm{\lambda}]$, $\mathcal{F}=\mathcal{L}[V]$ and $\mathcal{G}$ be the set of finite subsets of $\mathbb{Z}_{2^d}[V]$.
A map $\mathrm{PNF}: \mathcal{F}\times \mathcal{G}\rightarrow \mathcal{F}, (f, G)\mapsto \mathrm{PNF}(f\ |\ G)$ is called a \emph{parametric normal form} if for any $f\in \mathcal{F}$ and $G\in \mathcal{G}$, 
\begin{enumerate}
    \item $\mathrm{PNF}(0 \ |\ G) = 0$, 
    \item $\mathrm{PNF}(f \ |\ G) \neq 0$ implies $\lt(\mathrm{PNF}(f \ |\ G)) \notin \myangle{\{\lt(g)\}}_{\mathcal{L}[V]}$ for any $g\in G$, and
    \item $r:= f - \mathrm{PNF}(f \ |\ G)$ belongs to $\myangle{G}_{\mathcal{L}[V]}$.
\end{enumerate}
\end{definition}

In Algorithm~\ref{algo:param_normal_form_calc} (deferred to Appendix~\ref{appendix:param_normal_form_calc}), we show an instance of parametric normal form and its calculation.
Similar to Algorithm~\ref{alg:compute_groebner_bases}, at each time we attempt to find polynomial $g\in G$ satisfying its leading monomial divides the monomial of a term $t=c_t\cdot m_t$ of $f$.
A difference is that $c_t$ is no longer a constant and includes unknown parameters of $\bm{\lambda}$.
We consider the coefficients of these parameters.
Let $\bar{\nu}_2(c_t)$ be the minimum exponent of two of the coefficients.
For example, if $c_t = 2\lambda_1 + 4\lambda_2$, then $\bar{\nu}_2(c_t)= \min(\nu_2(2),\nu_2(4)) = 1$.
Hence, if $\bar{\nu}_2(c_t) \ge \nu_2(\lc(g))$, we scale $g$ and update $f$ by eliminating term $t$. 
The formal description of the normal form calculation and its soundness are deferred to Appendix~\ref{appendix:param_normal_form_calc}.
Below we apply the parametric normal form to check $\eta' - \mu\cdot \eta \in \myangle{F_\rho}$ with specific values of $\mu$ in the next proposition, whose proof is deferred to Appendix~\ref{appendix:proof_of_relation_of_param_nf_nf}.

\begin{proposition}
\label{prop:soundness_of_parametric_reduction}
For all $\mu\in \mathbb{Z}_{2^d}$ and concrete values $\overline{\bm{\lambda}}$ substituted for $\bm{\lambda}$, if 
 $\mathrm{PNF}(\eta' - \mu\cdot \eta\ |\ GB(F_\rho)) = 0$, then 
 $\mathrm{NF}(\eta' - \mu\cdot \eta\ |\ GB(F_\rho)) = 0$.
\end{proposition}

\begin{example}
Let $V = \{x,y\}$ and $\rho(V, V') = (x' = x+1) \wedge (y' = y+x)$ and the degree $k=2$.
Hence, the strong Gr\"obner basis of $F_\rho$ is $G= \{x'-x-1, y'-y-x\}$.
When $\mu$ is set as $1$, the parametric normal form of $\eta' - \mu\cdot \eta$ is shown in the left part below.
According to Proposition~\ref{prop:soundness_of_parametric_reduction}, to make $\eta' - \mu\cdot \eta \in \myangle{F_\rho}$, it suffices to solve the system of linear congruences shown in the right part. \hfill $\square$

\begin{center}
\begin{minipage}{0.55\textwidth} 
\begin{align*} 
&\eta\ =  \lambda_1 x^2 + \lambda_2 xy + \lambda_3 y^2 + \lambda_4 x + \lambda_5 y + \xi,\\
&\mathrm{PNF}(\eta' - \mu\cdot \eta\ |\ G) \\
&\quad =(\lambda_2 + \lambda_3)\cdot x^2 + 2\lambda_3\cdot xy + \\
&\quad \ \ \ (2\lambda_1 + \lambda_2 +\lambda_5)\cdot x + \lambda_2\cdot y  + (\lambda_1 +\lambda_4) 
\end{align*} 
\end{minipage}
\begin{minipage}{0.4\textwidth} 
$$
\Rightarrow\quad  \left\{
\begin{aligned}
&\lambda_2 + \lambda_3 = 0 \\
& 2\lambda_3 = 0\\
&2\lambda_1+ \lambda_2 + \lambda_5= 0 \\
& \lambda_2 = 0 \\
& \lambda_1 + \lambda_4 = 0
\end{aligned}
\right.
$$
\end{minipage}
\end{center}
\end{example}

\smallskip
\noindent{\em $\blacktriangleright$ Step B3. Finding values of parameters.} 
It can be verified that each coefficient of the returned value of Algorithm~\ref{algo:param_normal_form_calc} is a linear polynomial of variables of $\{\lambda_q\}_{q\in M^{(k)}[V]}$.
Thus, the normal form computation results in a system of linear congruences modulo $2^d$ over $\{\lambda_q\}_{q\in M^{(k)}[V]}$, for which a succinct representation of the set of solutions can be computed efficiently~\cite{DBLP:journals/toplas/Muller-OlmS07}. 
Each solution corresponds to a synthesized polynomial invariant.  
However, to verify or refute the postcondition, one may need to enumerate the solutions that can be exponentially many. 
To avoid the explicit enumeration, we take the strategy of first finding a specific solution $\bm{\lambda_0}$ to the system of congruences and then solving the congruences over the real number space $\mathbb{R}$.
Let $\mathrm{Null}(\mathbf{A})$ be the nullspace of the coefficient matrix $\mathbf{A}$ over $\mathbb{R}$.
Then, we shift the nullspace through $\mathrm{Null}(\mathbf{A}) + \bm{\lambda_0}$ and use them to synthesize invariants.
A detailed description can be found in Appendix~\ref{appendix:description_of_methods_for_find_sol}.

In this work, we allow the initial condition to have inequalities.  
Inequalities cause the problem that the value of the unknown coefficient $\xi$ may not be determined when $\mu=1$.
This is because $\xi$ is eliminated in $\eta' - \mu\cdot \eta$, and we may not be able to determine the value of $\xi$ using the initial condition.
To tackle the inequalities, we add a fresh variable $x_0$ for each program variable $x\in V$.
Meanwhile, we replace every variable $x$ by $x_0$ in the $\theta(V)$.
Then $\xi$ can be expressed by a polynomial expression in the variables of $\{x_0: x\in V\}$.
For example, consider Example~\ref{example:smt_solving_for_inv} but we relax the initial condition to $\theta(V) = (0<x<10) \wedge (0 <y < 10)$.
We observe that it is impossible to give a concrete value of $\xi$ such that  $x- y + \xi = 0$ is always guaranteed during the loop since we do not know their initial values.
We introduce two fresh variables $x_0, y_0$, for the initial values of $x$ and $y$, respectively.
Then we express $\xi$ by $x_0 - y_0$ and add constraints on $x_0$ and $y_0$ based on the initial condition, i.e., $0 < x_0 , y_0 < 10$.

\smallskip
\noindent{\em $\blacktriangleright$ Step B4. Verifying pre- and postconditions.} 
Finally, with the solved invariant from the previous step, we translate the original CHC constraints into simpler forms that can be efficiently verified by existing SMT solvers.
We utilize SMT solvers to verify that the initialization constraint and the verification (or refutation) constraint hold.
We provide a formal description of our transition procedure in Appendix~\ref{appendix:formal_desc_of_step_b4}. The overall soundness is stated in the following theorem. 
\begin{theorem}[Soundness]
\label{thm:soundness_of_loop_inv_trans}
If our algorithm outputs satisfiable (or resp. unsatisfiable) results, then the postcondition in the input is correct (or resp. refuted). 
\end{theorem}

\section{Implementation and Experimental Results}

In this section, we present the experimental evaluation of our SMT-solving approach. 
We first describe the detailed implementation of our approach, then the evaluation of our approach in quantifier-free bit-vector solving, and finally, polynomial invariant generation. 
All the experiments are conducted on a Linux machine with a 16-core Intel 6226R @2.90GHz CPU.

\subsection{Implementation}

\smallskip
\noindent{\em Quantifier-free equational bit-vector theory.} We implement the computation of strong Gr\"{o}bner bases with our improvement for calculating multiplicative inverse in the SMT solver \texttt{cvc5}~\cite{barbosa2022cvc5}. 
We use two computer algebra systems, namely \textsf{CoCoA-5}~\cite{CocoaSystem} and \textsf{Maple}~\cite{maple}.
In detail, we first parse the input formula in \textsf{cvc5} format to a polynomial in \textsf{CoCoA-5}, then implement the computation for strong Gr\"obner bases (Algorithm~\ref{alg:compute_groebner_bases}) in \textsf{CoCoA-5} to check the emptiness (Step A2), and finally resort to our complete decision procedure (Step A3) to find a solution where we use  \textsf{Maple} to find the roots of a univariate polynomial in $\mathbb{Z}_{2^n}$.

\smallskip
\noindent{\em Polynomial invariant generation.}
We implement our approach in Python 3.8.10.
Our implementation accepts inputs in the SMT-lib CHC format (\ref{eq:chc_constraint}) that include the initial, consecution, and postconditions, generates polynomial invariants by the algorithm given in Section~\ref{sec:poly_loop_inv} with $\mu$ set as $-1, 0, 1$, and uses \textsf{z3} to check whether the postcondition is verified or refuted and whether the initial condition implies the loop invariant. 
In the implementation, we invoke the Python package \textsf{Diophantine} for finding solutions of systems of linear congruences and \textsf{Sympy} for finding the nullspace of a matrix. 

\subsection{Quantifier-Free Equational Bit-Vector Theory}

\begin{table}[t]
\centering
\caption{Number of solved and unknown instances of all approaches.}
\begin{tabular}{|l|ccccc|c|}
\hline
\label{tab:num_solved}
\textbf{Solvers} & sat & unsat & unknown &timeout (10 sec) & memout (1 gb) & Solved \\ \hline
Bitwuzla & 674 & 354   & 281 & 4     & 1153  &  1028 \\
z3  & 499    &   436  & 281 &   1117  &   133 &  935  \\
cvc5 & 418 &  223    & 281 &  101  &  1443   & 641  \\
MathSAT  & 583  & 376   &  0 &  143  &  1364  &  959  \\
{\bf cvc5-sgb}  & \textbf{631} & \textbf{620}  & \textbf{0} & \textbf{1112}         & \textbf{103} &       \textbf{1252}          \\
\hline
\textbf{In total\quad }                                      &   &  & & &  &  2466        \\
\hline
\end{tabular}		
\end{table}

\begin{figure}[t]
    \centering
   \begin{subfigure}[b]{0.45\textwidth}
     \centering
    \includegraphics[width=1.1\textwidth]{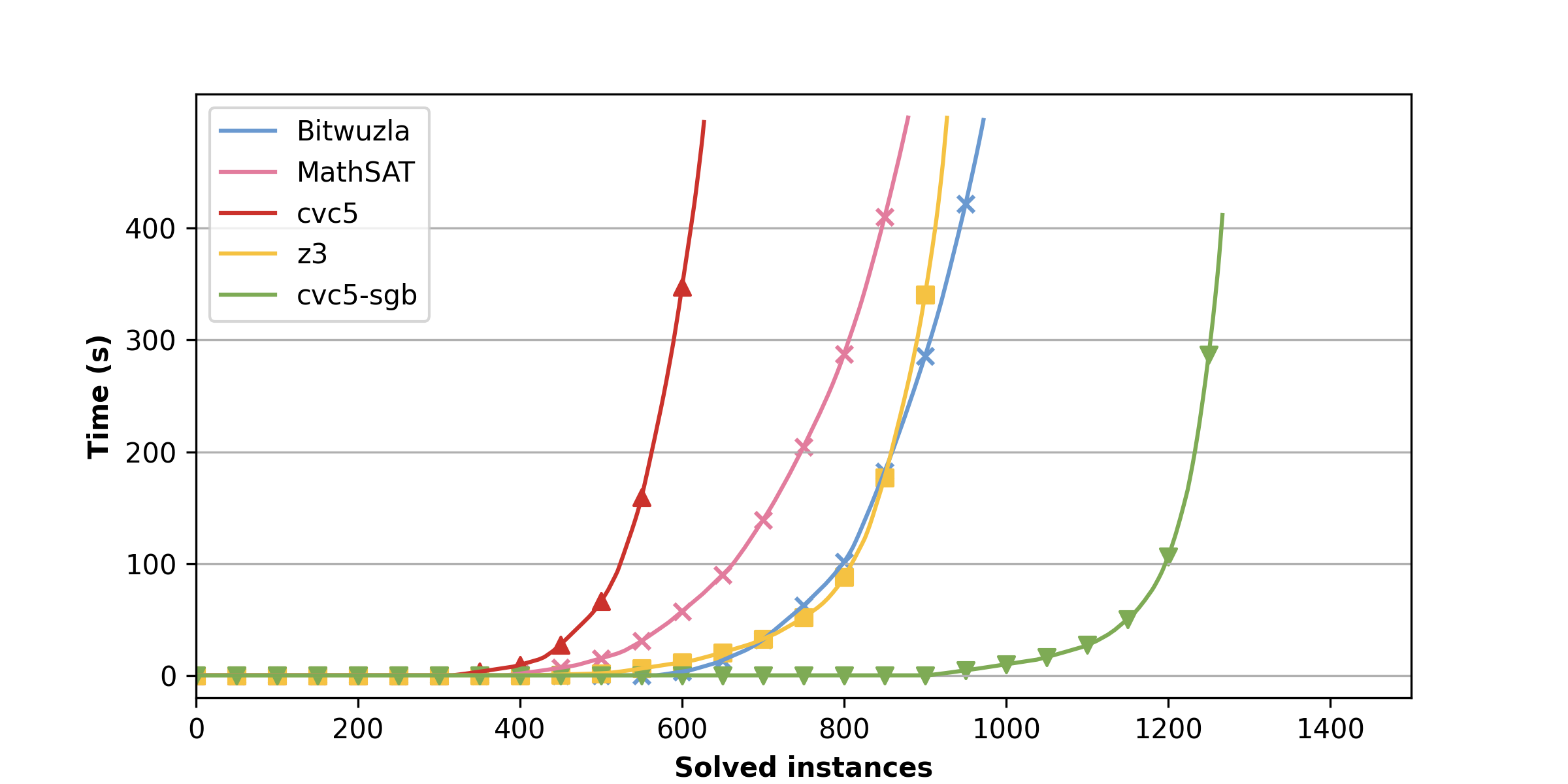}
    \caption{Comparison of Time Efficiency}
    \label{fig:performance_of_all_approaches}
  \end{subfigure}
  \begin{subfigure}[b]{0.45\textwidth}
     \centering
    \includegraphics[width=1.1\textwidth]{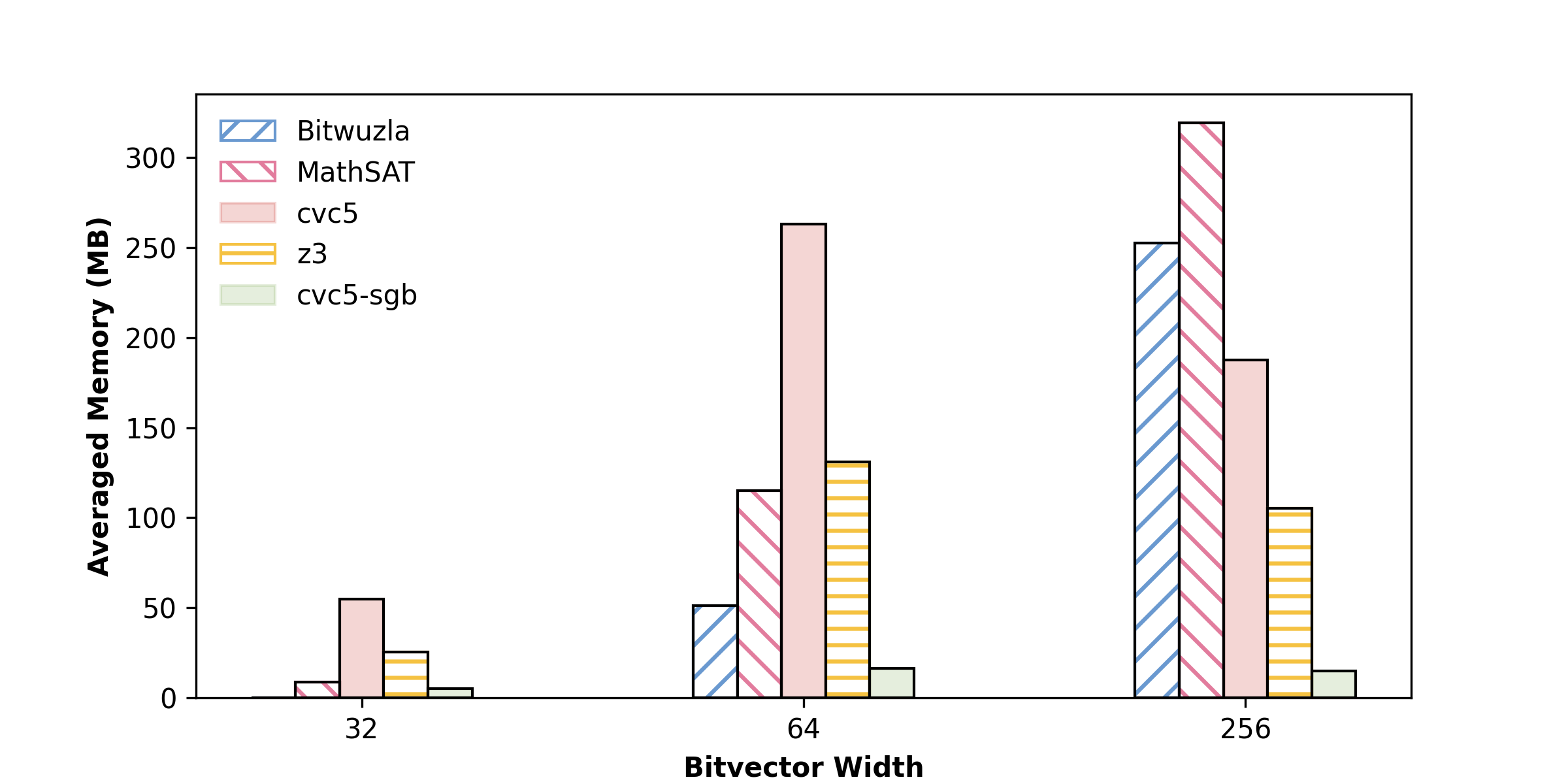}
    \caption{Comparison of Memory Usage}
    \label{fig:used_mem_of_all_approaches}
  \end{subfigure}
\end{figure}

\smallskip
\noindent{\em Benchmark setting.} We consider the  extensive benchmark set in~\cite{DBLP:conf/cav/OzdemirKTB23} as the baseline. The benchmarks in ~\cite{DBLP:conf/cav/OzdemirKTB23} are randomly generated quantifier-free formulas in a finite field $\mathbb{F}_p$ modulo a prime number $p$ . We adapt these benchmarks to \texttt{QF\_BV} (quantified-free, bit-vector) formulas as follows. 
First, if a variable (or constant) belongs to a finite field $\mathbb{F}_p$ satisfying $2^{d-1} < p < 2^d$, we then adapt it to a bit-vector variable with size $d$.
Second, the arithmetic operations over $\mathbb{F}_p$ are directly adapted to modular arithmetics over bit-vectors. 
The adapted benchmark set includes $2466$ benchmarks in SMT-LIB format and involves polynomial (in)equations of bit-vectors.

\smallskip
\noindent{\em Performance analysis.} 
We compare our methods with four state-of-art bit-vector solvers (\textsf{Bitwuzla}~\cite{DBLP:conf/cav/NiemetzP23}, \textsf{z3}~\cite{de2008z3}, \textsf{cvc5}~\cite{barbosa2022cvc5} and \textsf{MathSAT}~\cite{mathsat5}). 
\textsf{Bitwuzla} uses bit-blasting, \textsf{MathSAT} uses integer solving, and \textsf{z3}, \textsf{cvc5} are prominent comprehensive SMT solvers.
We set a time-out limit of $10$ seconds and a memory-out limit of $1$GB physical memory. 
Table~\ref{tab:num_solved} lists the number of solved instances of different solvers, where \textsf{cvc5-sgb} is our approach ($\S$\ref{sec:smt_acceleration}).
Our approach outperforms others in the number of solved instances (approx. $20\%$ more), especially in verifying unsatisfiable instances (approx. $40\%$ more) due to the use of strong Gr\"{o}bner bases.
It is worth noting that all the unsatisfiable instances are found by strong Gr\"{o}bner bases without invoking \textsf{Maple}.
Fig.~\ref{fig:performance_of_all_approaches} further shows the total time consumption on the solved instances of all approaches.
We observe that our approach has the lowest time consumption. 
Moreover, Fig.~\ref{fig:used_mem_of_all_approaches} depicts the average memory usage of the solved instances with different bit-vector sizes ($32, 64, 256$) by each approach. 
Our approach consumes the lowest amount of memory, while the memory consumption of others increases significantly as the size grows.

\subsection{Polynomial Invariant Generation}

\smallskip
\noindent{\em Benchmark setting.} We collect invariant generation tasks from three benchmark sets: 2016.Sygus-Comp~\cite{alur2016sygus}, 2018.SV-Comp~\cite{beyer2017software} and 2018.CHI-InvGame~\cite{10.1145/3173574.3173805}.
We classify the benchmarks into \emph{lin} ones that require only linear invariants and \emph{poly} ones that require polynomial invariants with degree more than one. There are in total $39$ linear benchmarks and $19$ poly benchmarks.

\begin{table}[t]
\centering
\caption{Performance of \textsf{Eldarica} and our approach on three datasets, where ``$\#$" is the number of benchmarks and the last line shows the averaged performance. 
} 
\begin{tabular}{|l|ccccc|crcrc|}
\hline
\multirow{3}{*}{\textbf{Dataset}} & \multicolumn{5}{c|}{\bf Linear}                                                                                                     & \multicolumn{5}{c|}{\bf Polynomial}                                                                                                        \\ \cline{2-11} 
& \multicolumn{1}{r|}{\multirow{2}{*}{$\#$}} & \multicolumn{2}{c|}{\bf Time (s)}                           & \multicolumn{2}{c|}{\bf Mem (mb)}    & \multicolumn{1}{r|}{\multirow{2}{*}{\#}} & \multicolumn{2}{c|}{\bf Time (s)}                           & \multicolumn{2}{c|}{\bf Mem (mb)}     \\ \cline{3-6} \cline{8-11} 
& \multicolumn{1}{r|}{}                    & \multicolumn{1}{r}{Eldarica} & \multicolumn{1}{r|}{Our} & \multicolumn{1}{r}{Eldarica} & Our & \multicolumn{1}{r|}{}                    & \multicolumn{1}{r}{Eldarica} & \multicolumn{1}{r|}{Our} & \multicolumn{1}{r}{Eldarica} & Our \\ \hline
2016.Sygus-Comp~\cite{alur2016sygus}  &   23 & 15.0 & \textbf{1.3} & 234.1 & \textbf{71.6} &   6  &  $>70.5$ & \textbf{5.5} & $>460.1$ &  \textbf{71.0}    \\
2018.SV-Comp~\cite{beyer2017software}  &   16 & 15.0 & \textbf{0.9} & 248.0 & \textbf{67.6} &   1  & $5.1$ & \textbf{0.5}  & $179.7$ & \textbf{70.6}     \\
2018.CHI-InvGame~\cite{10.1145/3173574.3173805}  &   0 & - & - & - & - &  12  & $>173.5$ & \textbf{7.4} & $>490.3$ & \textbf{72.5}     \\
\hline
\textbf{Average}   &   39 &  15.0 &  \textbf{1.1} &  239.8 & \textbf{70.0}  &   19  & $>132.1$  &  \textbf{6.4} & $>464.4$ & \textbf{71.9}   \\
\hline
\end{tabular}
\label{tab:performance_of_inv}
\end{table}

\smallskip
\noindent{\em Performance Analysis.} We compare our approach with \textsf{Eldarica}~\cite{8603013}, a state-of-art Horn clause solver that has the best performance in the experimental evaluation in \cite{DBLP:conf/kbse/YaoKSFWR23}. 
We set the time limit to $200$ seconds and the memory limit to $1$GB. Our approach solves all the benchmarks, while \textsf{Eldarica} times out over $11$ poly benchmarks. 
Table~\ref{tab:performance_of_inv} shows the average time and memory consumption over the benchmarks.
From the table, we can observe our approach has substantially better time efficiency and memory consumption, especially over poly benchmarks. 
For the polynomial benchmark, our approach has achieved an average speedup of more than $20$X in terms of time and an average reduction of more than $7$X in terms of memory.
Due to the space constraints, the detailed performance between \textsf{Eldarica} and our approach over these benchmarks is relegated to Table~\ref{tab:inv_on_2016_sygus_comp} and Table~\ref{tab:inv_on_2018_SV_comp} in Appendix~\ref{appendix:exper_results_of_loop_inv}, where ``$>200$" indicates that \textsf{Eldarica} times out on the benchmark. 
It is also worth noting that for most of these benchmarks, \textsf{z3} times out, and \textsf{cvc5} quickly outputs ``unknown''. 

\section{Related Works}
Bit-blasting approaches (\cite[Chapter 6]{DBLP:series/txtcs/KroeningS16}, \cite{DBLP:conf/issta/JiaH00MZ23}) decompose a bit-vector into the bits constituting it and solve the SMT problem by boolean satisfiability over these bits. Bit-blasting ignores the algebraic structure behind modular addition and multiplication
and hence cannot utilize them to speed up SMT solving. Compared with bit blasting, our approach leverages strong Gr\"{o}bner base to speed up the SMT solving of equational bit-vector theory.  

Integer-solving approaches~\cite{DBLP:conf/fmcad/Griggio11,DBLP:conf/vmcai/Jovanovic17,DBLP:conf/smt/Graham-Lengrand17} reduce bit-vector problems to the SMT solving of integer properties. 
Although bit-vectors can be viewed as a special case of integers, nonlinear integer theory is notoriously difficult to solve. 
Compared with these approaches, our approach solves the polynomial theory of bit-vectors via strong Gr\"{o}bner bases, which are more suitable to characterize algebraic properties of bit-vectors than general nonlinear integer theory. 

The application of Gr\"{o}bner bases to the equational theory of bit-vectors has been considered in previous works such as~\cite{DBLP:conf/fmcad/KaufmannBK19,JPAAgroebnerbasis,Pavlenko2011STABLEAN}. 
The Gr\"{o}bner bases used in these approaches are either restrictive (such as \cite{DBLP:conf/fmcad/KaufmannBK19} that uses Gr\"{o}bner bases over principal ideal domains) or too general (such as~\cite{JPAAgroebnerbasis,Pavlenko2011STABLEAN} that uses Gr\"{o}bner bases over general commutative Noetherian rings). 
This results in excessive computation to derive a Gr\"{o}bner basis. 
Compared with these results, our approach uses the strong Gr\"{o}bner bases that lie between principal ideal domains and commutative Noetherian rings and avoids the excessive computation by the succinct strong reduction in computing a strong Gr\"{o}bner basis.
Some recent work~\cite{DBLP:conf/cav/WienandWSKG08,DBLP:conf/sat/SeedKE20} adopts the methods of~\cite{JPAAgroebnerbasis} to compute a Gr\"obner basis without the need for solving an integer programming problem each time.
However, we have observed the Gr\"obner bases used in~\cite{DBLP:conf/cav/WienandWSKG08} are equivalent to strong Gr\"obner bases presented in~\cite{norton2001strong} while that in~\cite {DBLP:conf/sat/SeedKE20} is actually a slightly improved version of the strong Gr\"obner bases.

It is also worth noting that the work~\cite{DBLP:conf/fmcad/KaufmannBK19} applies Gr\"{o}bner bases to the special case of acyclic circuits, the work~\cite{JPAAgroebnerbasis} focuses on the special case of boolean fields, and the work~\cite{DBLP:conf/sat/SeedKE20} uses Gr\"obner bases for bit-sequence propagation, which solves the SMT problem by alternative methods.
In contrast, our approach completely uses strong Gr\"{o}bner bases to SMT solving and tackles invariant generation. 

Abstract interpretation~\cite{DBLP:conf/sas/SeedCKE23,DBLP:journals/toplas/Muller-OlmS07,DBLP:journals/toplas/ElderLSAR14,DBLP:conf/popl/CousotC77} has also been considered to generate polynomial equational invariants over bit-vectors. 
These approaches rely on well-established abstract domains and, hence, are orthogonal to our approach. 

Gr\"{o}bner bases have also been applied to finite fields~\cite{DBLP:conf/cav/OzdemirKTB23} and real numbers~\cite{DBLP:conf/popl/SankaranarayananSM04,cachera2014inference,rodriguez2004automatic}. 
The work~\cite{DBLP:conf/cav/OzdemirKTB23} utilizes Gr\"obner bases to determine the satisfiability of polynomial (in)equations where coefficients come from a finite field $\mathbb{F}_p$, while the work~\cite{rodriguez2004automatic,DBLP:conf/popl/SankaranarayananSM04,cachera2014inference} explores the generation of real polynomial loop invariants via Gr\"{o}bner bases.
Compared with these results, our approach is orthogonal as we consider polynomials over finite rings and use strong Gr\"{o}bner bases. 

\section{Conclusion and Future Work}
We have proposed a novel approach for SMT solving of equational bit-vector theory via strong Gr\"{o}bner bases, including the quantifier-free case and the quantified case of invariant generation. 
One future work is to optimize our approach using methods in e.g. \cite{DBLP:conf/sas/SeidlFP08}. 
Another direction is to extend our methods to more bit-vector arithmetics, such as bitwise operations and inequalities by, e.g., the Rabinowitsch trick~\cite[Chapter 4.2, Proposition 8]{cox2013ideals},\cite{DBLP:conf/lpar/HaderRK23}.
It is also interesting to incorporate our approach with program synthesis techniques such as~\cite{DBLP:journals/pacmpl/ParkDR23}.

\bibliographystyle{splncs04}
\bibliography{reference}

\newpage
\appendix
\section{Omitted Content of Section~\ref{sec:smt_acceleration}}

\subsection{Omitted Proof of Proposition~\ref{prop:odd_inverse}}
\label{appendix:omitted_proof_of_prop1}

\begin{proof}
First, let $\nu_2(a)= \alpha$.
Then $a$ can be written in the form of $a = 2^\alpha\cdot a'$, where $a'$ is an odd number.
Thus, according to Bézout’s lemma, there exist $s, t\in \mathbb{Z}_{2^d}$ such that $s\cdot a' + t\cdot 2^d = 1$. 
Thus, $a\cdot s = 2^\alpha \cdot (1- t\cdot 2^d) = 2^\alpha$.

Second, we show when $\alpha=0$, $b$ is unique.
Otherwise, if there are two distinct $b_1, b_2$ such that $ab_1 = ab_2 =1$, then $a\cdot (b_1 - b_2) = 0$. 
Since $a$ is odd and $b_1 \neq b_2$, it causes a contradiction.
\end{proof}

\subsection{Omitted Parts of Proposition~\ref{prop:equiv_cond_for_phi_sat}}
\label{appendix:omitted_proof_of_prop2}
\begin{proof}
First, we show that $f_i \neq g_i$ has a solution is equivalent to $e_i(f_i-g_i) = 2^{d-1}$ has a solution.
If $f_i\neq g_i$ has a solution, then there exists a non-zero constant number $c\in \mathbb{Z}_{2^d}$ such that $f_i- g_i = c$ has solutions.
Since $c\neq 0$, according to Proposition~\ref{prop:odd_inverse}, there exists an integer $b\in \mathbb{Z}_{2^d}$ such that $cb = 2^{d-1}$. 
Thus, $e_i(f_i - g_i) = 2^{d - 1}$ definitely has a solution.
On the other hand, it is not hard to verify each solution of $e_i(f_i - g_i) = 2^{d-1}$ is also a solution of $f_i \neq g_i$. 
Thus, the above statement implies $\Phi$ is satisfiable if and only if $\bigwedge_{i} (h_i = 0)$ is satisfiable, which is equivalent to the variety of $H$ is non-empty.
\end{proof}

\subsection{Omitted Proof of Theorem~\ref{prop:groeb_cond_over_gr}}
\label{appendix:omitted_proof_of_prop:groeb_cond_over_gr}
\begin{proof}
To begin with, we show that $\apoly{f}$ and $\spoly{f, g}$ satisfy the definitions of A-polynomial and S-polynomial, respectively.
First, for each $a\in \mathrm{Ann}(\lc(f))$, since $a \cdot \lc(f) = 0$, $\nu_2(a) \ge d - \nu_2(\lc(f))$.
On the other hand, it is not hard to verify each integer $b\in \myangle{2^{d-\nu_2(\lc(f))}}_{\mathbb{Z}_{2^n}}$ belongs to  $\mathrm{Ann}(\lc(f))$
Thus, $\myangle{2^{d-\nu_2(\lc(f))}}_R =  \mathrm{Ann}(\lc(f))$, which implies $\apoly{f} \in \mathrm{Apoly}(f)$.
Second, since $c_1\lc(f) = c_2\lc(g) = 2^v$ is definitely a least common multiple of $\lc(f)$ and $\lc(g)$, $\apoly{f,g}$ is a valid S-polynomial of $f$ and $g$.

Next, as ~\cite[Proposition 3.9, 6.2]{norton2001strong}, we know Algorithm~\ref{alg:compute_groebner_bases} always terminates and returns a strong Gröbner basis when the coefficient ring is $\mathbb{Z}_{2^n}$.
\end{proof}

\subsection{Omitted Proof of Theorem~\ref{thm:correctness_of_find_zeros}}
\label{appendix:omitted_proof_of_correctness_of_find_zeros}
\begin{proof}
First, we show the complete procedure always terminates by demonstrating the operations of Line~\ref{line:solve_univar},\ref{line:factorization}, and \ref{line:exhaust} can only be executed finite times.
Each time we execute Line~\ref{line:solve_univar}, the number of unassigned variables in $M$ strictly decreases.
Hence, it cannot be executed infinite times.
In addition, since each polynomial can only be factorized into finite polynomials over $\mathbb{Z}$ whose degrees are non-zero and the range of $i$ is limited, the operation in Line~\ref{line:factorization} will also be executed finite times.
Moreover, since the size of $\mathbb{Z}_{2^n}$ is finite, the last operation (in Line~\ref{line:exhaust}) will also be executed finite times.

Second, we show it will always return a valid solution if $\mathcal{V}(H)$ is non-empty. 
When $G$ contains a univariate polynomial $p\in \mathbb{Z}_{2^n}[x_i]$, for any solution $z\in \mathcal{V}(G)$, it also belongs to $Zeros(p)$.
Second, if $p$ can be decomposed to two polynomials $f, g$ over $\mathbb{Z}$,  for any solution $z\in \mathcal{V}(G)$, it would also be a solution to $(f\cdot 2^i = 0) \wedge (g\cdot 2^{d-i} = 0)$ for some $0\le i \le d-1$.
Conversely, for each $z$ of $\mathcal{V}(G\cup \{f\cdot 2^i, g\cdot 2^{d-i}\}\setminus \{p\})$, it also belongs to $\mathcal{V}(G)$.
Finally, the correctness of the exhausted search is obvious.
\end{proof}

\subsection{Omitted Description of Finding Multiplicative Inverse}
\label{appendix:proof_of_complexity_of_mul_inverse}

Before presenting the proof of Theorem~\ref{thm:complexity_of_mul_inverse}, we first present our approach to implementing the operation $f\leftarrow \frac{2^d-r}{a}$ (in Line~\ref{line:div_operation} of Algorithm~\ref{alg:find_inverse}) through Newton's method.
Then we will estimate the complexity of the number of arithmetic operations and binary operations. 

Since $1\le r \le a-1$, we have $\frac{2^d-r}{a} = \left\lfloor \frac{2^d}{a} \right\rfloor$.
To calculate it, we first give an estimation of $\frac1a$ through Newton's method and then multiply it by $2^d$ through one bitwise shift operation.
Next, consider the binary representation of $\frac{1}{a}$.
If it is finite, it can have a non-zero bit at most in the first $\lceil\log(a)\rceil$ bits.
Otherwise, since $a$ is odd, it could be represented as $0.\overline{(b_1\ldots b_\ell)}_2$, where $b_1\ldots b_\ell$ is the repetend with $\ell \le a-1$~\cite{dickson1920history}.
Hence, we first calculate the first $2a$ bits of $a$, then determine the length of repetend, and finally extend it to $d$ bits.
Then, through a shifting operation (i.e., $\frac1a \gg d$), we are able to get the value of $\frac{2^d-r}a$.

We adopt Newton's method to find the first $2a$ bits of $\frac{1}a$. 
Let $f(x) = \frac{1}x - a$.
In particular, let $x_n$ be the approximate result of the zero of $f(x)$ at the $n$-th round. 
Initially, we let $ x_0 = \frac{1}{2^{\lceil\log a \rceil}}$.
We calculate it by the following equation,
\begin{equation}\label{eqn:recur_eqn}
x_{n+1} = x_n - \frac{f(x_n)}{f'(x_n)} = x_n - \frac{\frac{1}{x_n} - a}{-\frac{1}{x_n^2}}= x_n(2-x_n\cdot a)\,.    
\end{equation}
Hence, 
$$
\abs{x_n - \frac1a} = a\abs{x_{n-1} - \frac1a}^2 = \cdots = a^{2^n-1} \abs{x_0 - \frac1a}^{2^n} < (ax_0 - 1)^{2^n}\,. 
$$
Since the first bit of the fractional part of $\frac{a}{2^{\lceil\log a \rceil}}$ should be one, $\abs{ax_0-1} \le \frac12$.
To make $\abs{x_n - \frac1a} < \frac{1}{2^{2a}}$, it suffices to let $n \ge \log (2a) =  \log a + 1$. 

Afterward, we are going to find the concrete value for $\ell$.
We start from $i = 1$ to determine whether $b_1\ldots b_i$ is a valid repetend by shifting $b_1\ldots b_i$ and determining whether the first $2a$ bits of $b_1\ldots b_ib_1\ldots b_i\ldots$ and the fractional part of $x_n$ are the same. 
If not, we increase $i$ by one and continue the above procedure.
Its correctness is guaranteed by the following lemma.
\begin{lemma}
The above procedure returns a valid and minimum repetend of $\frac{1}{a}$.
\end{lemma}
\begin{proof}
Suppose the length of repetend returned by the above procedure is $i$. 
For the sake of contradiction, we assume the length of minimum repetend is $\ell$ satisfying $i < \ell < a$.
If $i$ divides $\ell$, that means the valid repetend $b_1\ldots b_\ell$ can also be written as the form of $b_1\ldots b_ib_1\ldots b_i\ldots$, which causes a contradiction.
Otherwise, since the first $2a$ bits of $b_1\ldots b_ib_1\ldots b_i\ldots$ and the fractional part of $x_n$ are the same, we have 
$$b_j = b_{\ell+j} = b_{(\ell + j) \ \mathrm{mod}\ i}, \text{for any } 1\le j \le i\,.$$
Since $\ell \mod i \neq 0$ is smaller than $i$, the above equation implies $\ell \mod i$ should be first returned by our procedure instead of $i$.
This causes a contradiction.
Hence, $b_1\ldots b_i$ is a valid repetend. 
Additionally, it is not hard to verify that it is also the minimum.
\end{proof}

Next, we show the proof of Theorem~\ref{thm:complexity_of_mul_inverse}, which estimates the number of arithmetic operations and binary operations.
\begin{proof}
In Line~\ref{line:2_d_mod_a}, we invoke the modular exponentiation algorithm.
In each round, we do one multiplication ($r\leftarrow r^2$) and one modulo operation ($r \leftarrow r \mod a$). 
There are at most $\lceil \log d \rceil$ rounds.
Hence, the number of arithmetic operations of Line~\ref{line:2_d_mod_a} is at most $2\lceil \log d \rceil$.

In Line~\ref{line:div_operation}, as described, we first adopt Newton's method to compute the first $4a$ bits of $\frac{1}a$. 
It runs no more than $\log a+1$ rounds, and in each round, there are two multiplication and one subtraction.
One multiplication costs no more than $\mathcal{O}(2\cdot (2a)^2)$ and a subtraction costs $\mathcal{O}(2\cdot 2a)$ binary operations.
Moreover, the latter procedure for finding $\ell$ costs no arithmetic operations but $\mathcal{O}(2a\cdot a)$ binary operations.
After that, we shift $x_n$ multiply and find its first $d$ bits, which costs $\mathcal{O}(d)$ binary operations.

In Line~\ref{line:extended_euclidean}, we invoke extended Euclidean algorithm to find $k_1, k_2$ such that $k_1\cdot r+ k_2\cdot a = -1$, where $r < a$.
It runs no more than $\mathcal{O}(\log a)$ rounds.
In each round, there are five arithmetic operations (including one division, two multiplications, and two substractions). 
Thus, this step costs $\mathcal{O}(5\cdot \log a)$ arithmetic operations.
Since each operand is no more than $a$, it costs $\mathcal{O}(\log a\cdot (6\log^2 a + 4\log a))$ binary operations. 

The multiplication in Line~\ref{line:return_final_res} is the most expensive operation.
It costs one multiplication and one subtraction, and $\mathcal{O}(2d^2 + 3d)$ binary operations.

By summing them up, we could get the number of arithmetic operations by
$$
\mathcal{O}\left(2\log d + 3(\log a+1) + 5\log a  + 3 \right) = \mathcal{O}(2\log d + 8\log a),
$$
and the number of binary operations by 
\begin{gather*}
\mathcal{O}(\log d\cdot ((\log a)^2 + \log a) + (\log a+1)\cdot (8 a^2 + 4 a) + 2 a^2 + d \\
 + \log a\cdot (6\log^2 a + 4\log a)  + 2d^2 + 3d) \\
= \mathcal{O}(2d^2 + 4d + a^2\cdot (8\log a  + 2) + 4a\cdot \log a),
\end{gather*}
which concludes Theorem~\ref{thm:complexity_of_mul_inverse}.
\end{proof}

\begin{remark}
When the first $2a$ bits of $\frac1a$ are precomputed, and the procedure of finding minimum repetend is optimized by the Knuth–Morris–Pratt (KMP) algorithm, the number of arithmetic operations can be optimized to 
$\mathcal{O}(2\log d + 5\log a +3)$.
Meanwhile, the number of binary operations can be optimized to 
\begin{gather*}
\mathcal{O}(\log d\cdot ((\log a)^2 + \log a) + a + d
+ \log a\cdot (6\log^2 a + 4\log a)  + 2d^2 + 3d) \\
= \mathcal{O}(2d^2 + 4d + a + 6\log^3 a)\,.
\end{gather*}
\end{remark}

\section{Omitted Proofs of Section~\ref{sec:loop_invariant}}

\subsection{Paramertic Normal Form Calculation}
\label{appendix:param_normal_form_calc}
\begin{algorithm}[H]
\caption{Paramertic Normal form of $f$ with respect to $G$.}
\label{algo:param_normal_form_calc}
\begin{algorithmic}[1]
\renewcommand{\algorithmicrequire}{\textbf{Parameters:}}	
\Require $f$ and $G$, where $f\in \mathcal{L}[V]$ and $\mathcal{L} = \mathbb{Z}_{2^n}[\bm{\lambda}]$ and $G\subseteq \mathbb{Z}_{2^n}[V]$
\Ensure  $\mathrm{PNF}_G(f)$ 
\Function{$\mathrm{PNF}_G(f)$}{}
\While{there exists $g\in G$ such that $\lm(g)\ |\ t$ and $\nu_2(\lc(g)) \le \bar{\nu}_2(c_t)$, where $t =c_tm_t$ is a non-zero term of $f$}\label{line:loop_cond}
\State $f \leftarrow f - \frac{c_t}{2^{\nu_2(\lc(g))}} \cdot \left(\frac{\lc(g)}{2^{\nu_2(\lc(g))}}\right)^{-1}\frac{m_t}{\lm(g)}g$\;
\EndWhile
\Return{$f$}\;
\EndFunction
\end{algorithmic}
\end{algorithm}

Given a polynomial $f$, suppose the returned value of Algorithm~\ref{algo:param_normal_form_calc} is $\mathrm{PNF}_G(f)$.
The correctness of Algorithm~\ref{algo:param_normal_form_calc} is shown as follows.
\begin{lemma}[Correctness of Algorithm~\ref{algo:param_normal_form_calc}]
Algorithm~\ref{algo:param_normal_form_calc} always terminates, and  $\mathrm{PNF}_G(f)$ is a valid parametric normal form.
\end{lemma}
\begin{proof}
We first show it always terminates.
Otherwise, suppose the leading monomial of input $f$ is $\bm{x}^{\bm{\alpha}}$ and the largest leading monomial of $G$ is $\bm{x}^{\bm{\beta}}$.
If the reduction algorithm does not terminate, that means there exists an infinite number of monomials between $\bm{x}^{\bm{\alpha}}$ and $\bm{x}^{\bm{\beta}}$, which is impossible since we are using well-ordered monomial ordering (see more details in~\cite[Chapter 1.4]{groebnerbasis}).

Then we prove $\mathrm{PNF}_G(f)$ is a valid parametric normal form by verifying the three properties of Definition~\ref{def:param_normal_form}.
The first property is obvious since the loop condition (in Line~\ref{line:loop_cond}) is not met.
Next, in each round, we scale a polynomial $g$ of $G$ and eliminate the leading term of $f$.
Hence, by summing them up, we know $f- \mathrm{PNF}_G(f)$ can be written as a linear combination of polynomials of $G$, where coefficients come from $\mathcal{L}[V]$.
Thus, $r = f - \mathrm{PNF}_G(f)$ belongs to $\myangle{G}_{\mathcal{L}[V]}$.
Finally, we show the second condition is also satisfied by contradiction.
If it is not satisfied, there exits $g\in G$, such that $\lt(\mathrm{PNF}_G(f)) \in \myangle{\lt(g)}_{\mathcal{L}[V]}$.
Hence, $\lt(\mathrm{PNF}_G(f))$ could be written in the form of $h\cdot \lt(g)$, where $h\in \mathcal{L}[V]$.
Thus, $\lm(g)$ divides $\lm(f)$ and $\nu_2(\lc(g)) \le \nu_2(\lc(f))$, which meets the loop condition (in Line~\ref{line:loop_cond}) and causes a contradiction.
\end{proof}

\subsection{Proof of Proposition~\ref{prop:soundness_of_parametric_reduction}}
\label{appendix:proof_of_relation_of_param_nf_nf}
\begin{proof}
According to the third property of Definition~\ref{def:param_normal_form}, $r= (\eta' - \mu\cdot \eta) - \mathrm{PNF}(\eta' - \mu\cdot \eta\ |\ G)$ belongs to $\myangle{G: \mathbb{Z}_{2^d}[V]}_{\mathcal{L}[V]}$.
Without loss of generality, assume $r = h_1\cdot g_1 + \ldots + h_k \cdot g_k$, where $h_i \in \mathcal{L}[V], g_i \in G$.
By substituting the assignment of $\bm{\lambda}$ into $h_i$, we can find that $r$ belongs to $\myangle{G}$ since $h_i$ currently turns an element of $\mathbb{Z}_{2^n}[V]$.
Since $ \mathrm{PNF}(\eta' - \mu\cdot \eta\ |\ G) = 0$ under this assignment, we could find $\eta' - \mu\cdot \eta$ belongs to $\myangle{G}$.
Thus, we have $\eta' - \mu\cdot \eta \twoheadrightarrow_{GB(F_\rho)}^*0$ under that assignment, which implies $\mathrm{NF}(\eta' - \mu\cdot \eta\ |\ G) = 0$.
\end{proof}

\subsection{Description of Methods for Finding Solutions.}
\label{appendix:description_of_methods_for_find_sol}
Formally, let $\mathcal{C} = \{c_i\}_{i=1}^{\abs{\mathcal{C}}}$ be the set of coefficients of $\mathrm{PNF}_{GB(F_\rho)}(\eta' - \mu\cdot \eta)$. 
We first construct a system of congruences: $c_i = 0\ (\mathrm{mod}\ 2^d)$ for $i=1, \ldots \abs{\mathcal{C}}$.
Next, we first utilize Smith normal form~\cite{butson1955systems} to find a specific solution $\bm{\lambda_0}$ of this system.
Afterward, we introduce a set of fresh variables $\{v_i\}_{i=1}^{\abs{\mathcal{C}}}$ and convert the original system of linear congruences to $c_i - v_i\cdot 2^d = 0$ for $i=1, \ldots \abs{\mathcal{C}}$.
Let $\mathbf{A}$ be the coefficient matrix of $(\{\lambda_q\}_{q\in M^{(k)}[V]}, \{v_i\}_{1\le i\le \abs{\mathcal{C}}})$.

Then we compute the nullspace of $\mathbf{A}$ over $\mathbb{R}$. 
Denote it by $\mathrm{Null}(\mathbf{A}) = \{(\bar{\bm{\lambda}}, \bar{\bm{v}})\ :\ \mathbf{A}\cdot (\bm{\lambda},\bm{v})^{\mathrm{T}} = \mathbf{0}\}$, where $\bar{\bm{\lambda}}$  and $\bar{\bm{v}}$ are assignments of $\bm{\lambda}$ and $\bm{v}$, represtively.
Since $\mathbf{A}$ is a matrix over integer domain $\mathbb{Z}$, the elements of each vector belonging to $\mathrm{Null}(\mathbf{A})$ are rational numbers. 
For each $\bm{s}= (\bar{\bm{\lambda}}, \bar{\bm{v}})\in \mathrm{Null}(\mathbf{A})$, let $\mathcal{D}(\bm{s})$ denote the set of all the denominators of $\bm{s}$.
Let $\mathcal{S}^\mu$ be defined as follows.
$$
\mathcal{S}^\mu \triangleq \{\bm{\lambda}_0 + \bar{\bm{\lambda }}\cdot \mathrm{lcm}(\mathcal{D}(\bm{s})):\bm{s} = (\bar{\bm{\lambda}}, \bar{\bm{v}})\in \mathrm{Null}(\mathbf{A})\}.
$$
We could verify each vector $\bm{s}\in \mathcal{S}^\mu$ is a solution to $c_i - v_i\cdot 2^d = 0$ for any $c_i \in \mathcal{C}$.
By substituting the elements of $\bm{s}$ into $\{\lambda_q\}_{q\in M_k[V]}$, we are able to get a set of concrete loop invariants. 

\subsection{Formal Description of Step B4}
\label{appendix:formal_desc_of_step_b4}
Formally, we first construct a set of fresh variables $V_0 = \{x_{1,0},\ldots, x_{n,0}\}$ representing the initial values of all variables.
Then, we construct a conjunction of assertion $\mathcal{A}_1$ by substituting $x_i$ by $x_{i,0}$ in the initial condition.

Next, regarding the values of $\mu$ we select, we are going to construct $\mathcal{A}_2^\mu, \mathcal{A}_3^\mu$ according to $\mathcal{S}^\mu$ as follows.
If $\mu = 1$, for each solution $\bm{s} = (\bar{\bm \lambda}, \bar{\bm v})\in \mathcal{S}^1$, we create a fresh bit-vector variable $c_s$ and construct two constraints $\eta(V_0, k) + c_s = 0$ and $\eta(V, k)  + c_s = 0$ by instantiating $\{\lambda_q\}$ by $\bar{\bm \lambda}$ and replacing $\xi$ by $0$.
Then we append the two constraints to $\mathcal{A}_2^1$ and $\mathcal{A}_2^3$ respectively.
Finally, we check whether $\mathcal{A}_1^1\wedge \mathcal{A}_2^1 \wedge \mathcal{A}_3^1 \Rightarrow \kappa(V)$ is satisfiable.

Otherwise, if $\mu \neq 1$, for each $\bm{s} = (\bar{\bm \lambda}, \bar{\bm v})\in \mathcal{S}^\mu$, we construct two constraints $\eta(V_0, k) = 0$ and $\eta(V, k) = 0$ by instantiating $\bm{\lambda}$ by $\bar{\bm \lambda}$, and append them to $\mathcal{A}_2^\mu$ and $\mathcal{A}_3^\mu$, respectively. 
Then we first check whether the calculated invariant holds for the initial conditions by verifying $\theta(V_0)\Rightarrow \eta(V_0, k)$.
If yes, then we further check whether $\eta(V)\Rightarrow \kappa(V)$ holds.

For example, when the initial condition of Example~\ref{example:smt_solving_for_inv} is modified to $(0 < x < 10) \wedge (0 < y < 10)$, we are able to get the solution set $\mathcal{S}^1 = \{(\lambda_x:1, \lambda_y:-1)\}$ by setting $\mu = 1$ and $k =1$.
Then according to the above method, we can verify the postcondition  $\kappa(V) = (x_2 - x_1 < 10)$ through determining the satisfiablity of the below formula.

\begin{minipage}{0.65 \textwidth}
\begin{gather}
((0 < x_0 < 10) \wedge (0 < y_0 < 10) \tag{$\mathcal{A}_1$} \\ 
\wedge (x_0 - y_0 + c_s =0 )  \tag{$\mathcal{A}_2^{1}$} \\ 
\wedge(x - y + c_s = 0) \tag{$\mathcal{A}_3^{1}$})    
\end{gather}
\end{minipage}
\begin{minipage}{0.3\textwidth}
$$
\Rightarrow \quad (y - x< 10)\,.
$$
\end{minipage}

\section{Experimental Results of Loop Invariant Generation}
\label{appendix:exper_results_of_loop_inv}         

\begin{table}[ht]
\centering
\caption{Performance of \textsf{Eldarica} and our method in 2016.Sygus-Comp~\cite{alur2016sygus}.}
\label{tab:inv_on_2016_sygus_comp}
\begin{tabular}{|l|l|rr|rrrr|}
\hline
\multirow{2}{*}{\textbf{2016.Sygus-Comp~\cite{alur2016sygus}}} & \multirow{2}{*}{\bf Inv} & \multicolumn{2}{c|}{\bf Eldarica~\cite{8603013}}      & \multicolumn{4}{c|}{\bf Our Method $(\S$\ref{sec:loop_invariant})}                                                                   \\ \cline{3-8}
&   & \multicolumn{1}{c}{Time (s) } & Mem (MB) & \multicolumn{1}{c}{Time} & \multicolumn{1}{l}{Speedup} & \multicolumn{1}{l}{Mem} & Save up \\ \hline
\texttt{anfp-new} &                   \textsf{Poly}& $> 200$ & $>10^3$     & 0.5 & 62.5X & 71.0 & 5.4X         \\ 
\texttt{anfp} &                       \textsf{Poly}& $>200$  & $> 10^3$    & 0.3 & 107.1X& 70.2 & 4.9X         \\ 
\texttt{cegar1\_vars-new} & \textsf{Lin}           & 25.3    &  273.5      & 1.6 & 15.9X & 70.5 & 3.9X         \\ 
\texttt{cegar1\_vars} & \textsf{Lin}               & 26.8    &  289.2      & 1.5 & 18.3X & 70.4 & 4.1X         \\ 
\texttt{cegar1-new} & \textsf{Lin}                 & 29.4    &  279.2      & 1.0 & 30.0X & 70.3 & 4.0X         \\ 
\texttt{cegar1} & \textsf{Lin}                     & 24.4    &  266.5      & 1.0 & 24.7X & 70.2 & 3.8X         \\ 
\texttt{cgmmp-new} & \textsf{Lin}                  & 10.9    &  206.8      & 1.0 & 11.1X & 70.3 & 2.9X         \\ 
\texttt{cgmmp} & \textsf{Lin}                      & 11.1    &  212.5      & 0.9 & 12.4X & 73.2 & 2.9X         \\
\texttt{ex23\_vars} & \textsf{Lin}                 & 12.2    &  232.1      & 1.7 & 7.2X  & 70.7 & 3.3X         \\
\texttt{ex14\_simpl} & \textsf{Lin}                & 9.8     &  209.2      & 1.2 & 8.3X  & 70.6 & 3.0X         \\
\texttt{ex14\_vars} & \textsf{Lin}                 & 9.7     &  210.5      & 1.9 & 5.1X  & 70.8 & 3.0X         \\
\texttt{ex14-new} & \textsf{Lin}                   & 6.0     &  181.1      & 1.0 & 6.1X  & 70.3 & 2.6X         \\
\texttt{ex14} & \textsf{Lin}                       & 6.4     &  192.2      & 0.9 & 7.3X  & 70.5 & 2.7X         \\
\texttt{ex23} & \textsf{Lin}                       & 11.5    &  236.3      & 1.4 & 8.3X  & 70.6 & 3.3X         \\
\texttt{fig1\_vars-new} & \textsf{Poly}            & 5.9     &  195.7      &14.9 & 0.4X  & 71.8 & 2.7X         \\
\texttt{fig1\_vars} & \textsf{Poly}                & 5.9     &  199.4      &13.9 & 0.4X  & 71.7 & 2.8X         \\
\texttt{fig1-new} & \textsf{Poly}                  & 5.6     &  182.4      & 1.6 & 3.6X  & 70.6 & 2.6X         \\
\texttt{fig1} & \textsf{Poly}                      & 5.5     &  182.8      & 1.7 & 3.3X  & 70.9 & 2.6X         \\
\texttt{fig9\_vars} & \textsf{Lin}                 & 3.2     &  198.6      & 1.9 & 1.7X  & 70.2 & 2.8X       \\
\texttt{fig9} & \textsf{Lin}                       & 5.7     &  181.0      & 0.9 & 6.5X  & 60.9 & 3.0X        \\
\texttt{sum1\_vars} & \textsf{Lin}                 & 21.5    &  261.7      & 1.4 & 15.6X & 74.2 & 3.5X        \\
\texttt{sum1} & \textsf{Lin}                       & 20.0    &  258.4      & 1.4 & 14.3X & 74.0 & 3.5X        \\
\texttt{sum3\_vars} & \textsf{Lin}                 & 7.0     &  200.0      & 1.4 & 5.1X  & 74.3 & 2.7X         \\
\texttt{sum3} & \textsf{Lin}                       & 6.9     &  190.7      & 1.5 & 4.6X  & 74.3 & 2.6X        \\
\texttt{sum4\_simp} & \textsf{Lin}                 & 20.0    &  272.2      & 1.5 & 13.5X & 74.2 & 3.7X         \\
\texttt{sum4\_vars} & \textsf{Lin}                 & 22.0    &  256.5      & 1.3 & 17.2X & 73.9 & 3.5X        \\
\texttt{sum4} & \textsf{Lin}                       & 7.7     &  194.5      & 1.3 & 6.0X  & 73.9 & 2.6X        \\
\texttt{tacas\_vars} & \textsf{Lin}                & 22.1    &  272.4      & 1.6 & 14.0X & 74.6 & 3.7X         \\
\texttt{tacas} & \textsf{Lin}                      & 24.7    &  309.0      & 1.5 & 16.7X & 74.1 & 4.2X         \\
\hline
\textbf{$29$   in total}        & & & & & & &                                \\
\hline
\end{tabular}
\end{table}

\begin{table}[ht]
\centering
\subfloat[2018.SV-Comp~\cite{beyer2017software}]{
\begin{tabular}{|l|l|rr|rrrr|}
\hline
\multirow{2}{*}{\textbf{2018.SV-Comp~\cite{beyer2017software}}} & \multirow{2}{*}{\bf Inv} & \multicolumn{2}{c|}{\bf Eldarica~\cite{8603013}}      & \multicolumn{4}{c|}{\bf Our Method $(\S$\ref{sec:loop_invariant})}                                                                   \\ \cline{3-8}
&   & \multicolumn{1}{c}{Time (s) } & Mem (MB) & \multicolumn{1}{c}{Time} & \multicolumn{1}{l}{Speedup} & \multicolumn{1}{l}{Mem} & Save up \\ \hline
\texttt{cggmp2005\_true$\ldots$} & \textsf{Lin} & 11.9&   204.4    &  0.88 & 13.5X & 72.2 &  2.8X        \\ 
\texttt{cggmp2005\_variant$\ldots$} & \textsf{Lin}     & 28.4    &  349.9     &  0.98 & 30.6X & 72.3 & 4.8X         \\ 
\texttt{const-false-unreach-call1$\ldots$} & \textsf{Lin}     & 8.0    &  190.0     &  0.78 & 10.3X & 71.8 & 2.6X  \\ 
\texttt{const-true-unreach-call1$\ldots$} & \textsf{Lin}     & 11.2    &  206.1    &  0.78 & 14.3X & 71.8 & 2.9X  \\ 
\texttt{count\_up\_down\_$\ldots$} & \textsf{Lin}   & 6.8    &  196.2    &  0.98 & 6.9X & 72.2 & 2.7X  \\ 
\texttt{css2003\_true-unreach$\ldots$} & \textsf{Lin}   & 6.8    &  196.2    &  0.98 & 6.9X & 72.2 & 2.7X  \\ 
\texttt{down\_true-$\ldots$} & \textsf{Lin}   & 9.4   &  204.3    &  1.2 & 8.0X & 72.5 & 2.8X  \\ 
\texttt{gsv2008\_true-unreach$\ldots$} & \textsf{Poly}   & 5.1  &  179.7    &  0.5 & 10.5X & 70.6 & 2.5X  \\ 
\texttt{hhk2008\_true-unreach$\ldots$} & \textsf{Lin}   & 7.5  &  192.4   &  1.2 & 6.4X & 72.3 & 2.7X  \\ 
\texttt{jm2006\_true-unreach$\ldots$} & \textsf{Lin}   & 16.6 &  222.6   &  1.19 & 14.0X & 72.2 & 3.1X  \\
\texttt{jm2006\_variant\_true$\ldots$} & \textsf{Lin}   & 24.1 & 528.8 &  1.38 & 17.5X & 72.4 & 7.3X  \\
\texttt{multivar\_false-$\ldots$} & \textsf{Lin}   &  9.5 &  198.7  &  0.98 & 9.7X & 72.1 & 2.8X  \\
\texttt{multivar\_true-$\ldots$} & \textsf{Lin}   &  6.3 &  189.6  &  0.89 & 7.0X & 72.1 & 2.6X  \\
\texttt{simple\_$\ldots$-unreach-call1$\ldots$} & \textsf{Lin}   &  10.8 &  207.6  &  1.2 & 9.1X & 72.2 & 2.9X  \\
\texttt{simple\_$\ldots$-unreach-call2$\ldots$} & \textsf{Lin}   &  65.7 &  532.7  &  1.1 & 28.0X & 72.5 & 7.3X  \\
\texttt{while\_infinite\_loop\_3$\ldots$} & \textsf{Lin}   &  4.9 &  148.4  &  0.8 & 6.2X & 71.4 & 2.1X  \\
\texttt{while\_infinite\_loop\_4$\ldots$} & \textsf{Lin}   &  5.2 &  178.9 &  0.9 & 5.9X & 71.7 & 2.5X  \\
\hline
\textbf{$17$ in total}     & & & & & & &                                            \\
\hline
\end{tabular}
}
\qquad
\subfloat[2018.CHI-InvGame~\cite{10.1145/3173574.3173805}]{
\begin{tabular}{|l|c|rr|rrrr|}
\hline
\multirow{2}{*}{\textbf{2018.CHI\_InvGame~\cite{10.1145/3173574.3173805}}} & \multirow{2}{*}{\bf Inv} & \multicolumn{2}{c|}{\bf Eldarica~\cite{8603013}}      & \multicolumn{4}{c|}{\bf Our Method $(\S$\ref{sec:loop_invariant})}                                                                   \\ \cline{3-8}
&   & \multicolumn{1}{c}{Time (s) } & Mem (MB) & \multicolumn{1}{c}{Time} & \multicolumn{1}{l}{Speedup} & \multicolumn{1}{l}{Mem} & Saveup \\ \hline
\texttt{cube2.desugared} & \textsf{Poly}                      & $95.6$ &   $504.9$     & 19.3 & 5.0X & 75.0&  $6.7$X        \\ 
\texttt{gauss\_sum-more-rows.auto} & \textsf{Poly}           & $>200$ &   $>520.5$    & 0.78 & $>38.5$X & 70.9 &   $>7.3$X      \\ 
\texttt{s9.desugared} & \textsf{Poly}               & $>200$ &   $>676.9$     & 22.0 &  $>9.1$X & 73.2 &  $>9.2$X     \\ 
\texttt{s10.desugared} & \textsf{Poly}                 & $>200$ &   $> 512.0$     & $0.9$ & $>227.3$X & 71.4 & $>7.2$X         \\ 
\texttt{s11.desugared} & \textsf{Poly}      & $>200$ &$> 502.8$& 1.4 & $>145.9$X & $73.7$&$>6.8$X       \\ 
\texttt{sorin03.desugared} & \textsf{Poly}                  & $>200$    &  $>492.5$     & $1.8$ & $112.4$X & $71.2$ & $6.9$X         \\ 
\texttt{sorin04.desugared} & \textsf{Poly}                      & 75.5   &  417.5     & 19.5 & $3.9$X &73.0 & $5.7$X        \\
\texttt{sqrt-more-rows-swap-columns} & \textsf{Poly}               & $>200$ &   $>427.2$     & 12.2 & $16.4$X & 74.9 &  $>5.7$X       \\ 
\texttt{non-lin-ineq-1.desugared} & \textsf{Poly}              & 6.9    &  243.5     & 2.4 & $7.8$X & $71.0$ & $3.4$X         \\
\texttt{non-lin-ineq-3.desugared} & \textsf{Poly}               & $>200$ &   $>531.7$     & 2.0 & $>101.0$X & 71.9 &  $>7.4$X        \\ 
\texttt{non-lin-ineq-4.desugared} & \textsf{Poly}                    & $>200$ &   $> 525.2$     & 5.1 & $>39.3$X & 73.5 & $>7.1$X       \\ 
\texttt{s5auto.desugared} & \textsf{Poly}                     & $>200$ &   $>530.0$     & 0.9 & $>34.1$X & $70.8$ & $7.5$X        \\ 
\hline
\textbf{$12$ in total}                                         &  &     &        &  & & &   \\
\hline
\end{tabular}
}
\caption{Performance of \textsf{Eldarica} and our method on 2018.SV-Comp~\cite{beyer2017software} and 2018.CHI-InvGame~\cite{10.1145/3173574.3173805}.}
\label{tab:inv_on_2018_SV_comp}
\end{table}

\end{document}